%% file: ms.tex
\pdfoutput=1
\documentclass[10pt,a4paper]{article}
\def\version{arxiv}
\input{main-header.tex}

\newcommand{\quantcert}{\url{https://quantcert.github.io/Mermin-eval}}

\ifdeftostring{\version}{qip}
  {
\title{Mermin Polynomials for Entanglement Evaluation\\ 
in Grover's algorithm and Quantum Fourier Transform}
\author{Henri de Boutray$^{1,2}$ \and Hamza Jaffali$^{1,2}$ \and 
  Frédéric Holweck$^{1,3}$ \and Alain Giorgetti$^{1,2}$ \and 
  Pierre-Alain Masson$^{1,2}$}
\authorrunning{Henri de Boutray \and Hamza Jaffali \and Frédéric Holweck \and 
  Alain Giorgetti \and Pierre-Alain Masson}
\date{}
\institute{
  Corresponding author: henri.de-boutray[at]univ-fcomte.fr \at
  \and
  Affiliations: \at
  $\begin{array}{ll}
  ^1&\text{Univ. Bourgogne Franche-Comté (UBFC)}\\
  ^2&\text{Institut FEMTO-ST (UMR 6174 - CNRS/UBFC/UFC/ENSMM/UTBM)}\\
  ^3&\text{Laboratoire Interdisciplinaire Carnot de Bourgogne (ICB, UMR
  6303 - CNRS/UB/UTBM)}\\
  \end{array}$
}
  }

\ifdeftostring{\version}{arxiv}
  {
\title{Mermin Polynomials for Entanglement Evaluation\\ 
in Grover's algorithm and Quantum Fourier Transform}
\author[1,2]{Henri de Boutray}
\author[1,2]{Hamza Jaffali}
\author[1,3]{Frédéric Holweck}
\author[1,2]{Alain Giorgetti}
\author[1,2]{Pierre-Alain Masson}
\affil[1]{Univ. Bourgogne Franche-Comté (UBFC)}
\affil[2]{Institut FEMTO-ST (UMR 6174 - CNRS/UBFC/UFC/ENSMM/UTBM)}
\affil[3]{Laboratoire Interdisciplinaire Carnot de Bourgogne (ICB, UMR 6303)}
\date{}
  }

\begin{document}

\maketitle

\begin{abstract}
The entanglement of a quantum system can be valuated using Mermin polynomials.
This gives us a means to study entanglement evolution during the execution of
quantum algorithms. We first consider Grover's quantum search algorithm,
noticing that states during the algorithm are maximally entangled in the
direction of a single constant state, which allows us to search for a single
optimal Mermin operator and use it to evaluate entanglement through the whole
execution of Grover's algorithm. Then the Quantum Fourier Transform is also 
studied with Mermin polynomials. A different optimal Mermin operator is searched
at each execution step, since in this case there is no single direction of
evolution. The results for the Quantum Fourier Transform are compared to results
from a previous study of entanglement with Cayley hyperdeterminant. All our
computations can be replayed thanks to a structured and documented open-source
code that we provide. 
\ifdeftostring{\version}{arxiv}
  {
\blfootnote{Corresponding author: henri.de-boutray[at]univ-fcomte.fr}
  }
\end{abstract}


\newpage

\section{Introduction} 
\label{sec:introduction}

Quantum entanglement has been identified as a key ingredient in the speed-up of
quantum algorithms~\cite{JL03}, when compared to their classical counterparts.
Our work is in line with previous work on a deeper understanding of the
role of entanglement in this speed-up~\cite{EJ98,BP02,CBAK13,KM06}.

We focus on Grover's algorithm~\cite{Gro96} and the Quantum Fourier Transform
(QFT)~\cite[Chap.II-Sec.5] {NC10} which plays a key role in Shor's
algorithm~\cite{Sho94}. We choose these two examples because they both provide
quantum speed-up (quadratic for Grover and exponential for QFT) and are well
understood and described in the literature~\cite{NC10}. Previous work tackled
entanglement in Grover's algorithm and the QFT from two perspectives:
quantitatively, with the Geometric Measure of Entanglement (GME)~\cite{WG03},
separately for Grover's algorithm~\cite{RBM13} and the QFT~\cite{SSB05}, and
qualitatively, by observing the different entanglement SLOCC classes traversed
by an execution, for both algorithms~\cite{JH19}.

Instead of directly measuring entanglement we use Mermin
polynomials~\cite{Mer90,ACG+16,AL16} to demonstrate the non-local properties of
the states generated by these algorithms. As a generalization of the CHSH
inequalities~\cite{CHSH69}, Mermin polynomials have two advantages: the quantum
states' evaluation can be compared to a classical bound, and the evaluation has
a possible physical implementation. Batle et al.~\cite{BOF+16} previously
investigated non-local properties during Grover's algorithm using Mermin
polynomials. However they concluded to the absence of non-local phenomena. In
the present work we setup the Mermin polynomials in such a way that we exhibit,
on the contrary, violation of the classical inequalities in Grover's algorithm.
Moreover our evaluation techniques are more efficient, allowing us to reach 12
qubits. We also exhibit non-local behavior in the QFT.

After Section~\ref{sec:background} presenting some background on Grover's
algorithm, the QFT and Mermin polynomials, Section~\ref{sec:method_and_result}
presents our method and results concerning the evaluation of entanglement in
Grover's algorithm and the QFT. In particular we exhibit Mermin's inequalities
violations in both algorithms. In this section we also compare the results
obtained with the Mermin polynomials to previous results~\cite{JH19} using the
Cayley hyperdeterminant. Finally, Section~\ref{sec:implementation} documents the
code developed for this evaluation, in order to make it reusable by anyone
wishing to\footnote{The source code is available at \quantcert.}. In addition,
Appendix~\ref{sec:explicit_states_for_grover_s_algorithm} recalls known
properties of the states in Grover's algorithm and
Appendix~\ref{appendix:cayley_hyperdeterminant_D2222} recalls the definition of
the Cayley hyperdeterminant.


\section{Background} 
\label{sec:background}

This section provides the necessary background to the reader, regarding Grover's
algorithm~(\ref{sub:grover_algorithm}), some properties of the states during its
execution~(\ref{sub:properties_of_states_in_the_grover_algorithm}), the Quantum
Fourier Transform (\ref{sub:qft}) and the Mermin
operators~(\ref{sub:mermin_polynomials_and_mermin_inequalities}).

  \subsection{Grover's algorithm} 
  \label{sub:grover_algorithm}

We sum up here Grover's algorithm, widely described in the literature
(\cite{Gro96,LMP03} and~\cite[chapter 6]{NC10}).

Grover's algorithm aims to find objects satisfying a given condition in an
unsorted database of $2^n$ objects, \emph{i.e.} to solve the following problem.

\noindent\fbox{\parbox{\textwidth}{\emph{
Given a positive integer $n$, $N=2^n$, $\Omega = \llbracket 0, N-1\rrbracket$
and the characteristic function $f~: \Omega \rightarrow \{0,1\}$ of some subset
$S$ of $\Omega$ ($f(x) = 1$ iff $x \in S$), find in $\Omega$ an element of $S$
only by applying $f$ to some elements of $\Omega$.
}}}

\vspace{.5em}
Grover's algorithm provides a quadratic speedup over its classical counterparts.
Indeed, assuming that each application of $f$ is done in one step, it runs in 
$\mathcal{O}(\sqrt{N})$ instead of $\mathcal{O}(N)$.

\begin{figure}[hbt!]
$$
\Qcircuit @C=1em @R=.7em {
\lstick{\ket{0}} & /^n \qw & \multigate{1}{H^{\otimes n+1}} 
  & \multigate{1}{U_f} & \gate{\mathcal{D}} 
  & \qw & \cdots &
  & \multigate{1}{U_f} & \gate{\mathcal{D}}
  & \meter & \cw \\
\lstick{\ket{1}} & \qw & \ghost{H^{\otimes n+1}} 
  & \ghost{U_f} & \qw 
  & \qw & \cdots &
  & \ghost{U_f} & \qw 
  & \qw & \qw 
\gategroup{1}{4}{2}{5}{.7em}{--}
\gategroup{1}{9}{2}{10}{.7em}{--}
}
$$
\caption{Grover's algorithm in circuit formalism}
\label{fig:grover_circ}
\end{figure}
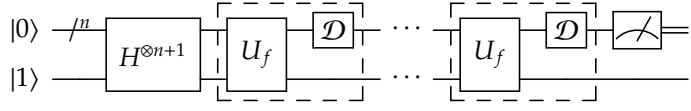

Figure~\ref{fig:grover_circ} shows this algorithm as a circuit composed of 
several gates that we now describe. $H^{\otimes n+1}$ is simply the Hadamard
gate on each wire. When applied on the $n$ first registers initialized at
$\ket{0}$, it computes the superposition of all states, \emph{i.e.},

$$H^{\otimes n}\ket{0}=\dfrac{1}{\sqrt{N}}\sum_{{\bf x}=0}^{N-1} \ket{\bf x}.$$

\noindent After $H^{\otimes n+1}$, the dashed box (hereafter called
$\mathcal{L}$) is repeated $k_{opt}=\left\lfloor\frac{\pi}{4} \sqrt{\frac{N}
{|S|}}\right\rfloor$ times. 

The circuit $\mathcal{L}$ is composed of the \emph{oracle} $U_f$ and the
\emph{diffusion operator} $\mathcal{D}$. The gate $U_f$ computes the classical
function $f$. It has the following effect on states:

$$\forall ({\bf x},y) \in \llbracket0,N\rrbracket\times\{0, 1\},\ U_f \left(
\ket{\bf x}\otimes\ket{y} \right) = \ket{\bf x}\otimes\ket{y \oplus f(\bf x)}.$$

\noindent On the circuit of Figure~\ref{fig:grover_circ} one can show that the
last register remains unchanged when applying the $U_f$ gate. Indeed after the
Hadamard gate $H$, this last register becomes $H \ket{1} = \dfrac{\ket{0} -
\ket{1}}{\sqrt{2}}$. Now consider a state $\ket{{\bf x}} \otimes \dfrac{\ket{0}
- \ket{1}}{\sqrt{2}}$. Then 

$$U_f~\left(\ket{\bf x} \otimes \dfrac{\ket{0} - \ket{1}}{\sqrt{2}}\right) =
\begin{cases}
\ket{\bf x} \otimes \dfrac{\ket{1}-\ket{0}}{\sqrt{2}} & \text{ if } f({\bf x})=1\\
\ket{\bf x} \otimes \dfrac{\ket{0}-\ket{1}}{\sqrt{2}} & \text{otherwise.}
\end{cases}$$

\noindent In other words, 

$$U_f~\left(\ket{\bf x} \otimes \dfrac{\ket{0} - \ket{1}}{\sqrt{2}}\right) =
(-1)^{f({\bf x})}~\left(\ket{\bf x} \otimes \dfrac{\ket{0} -
\ket{1}}{\sqrt{2}}\right).$$

\noindent  One says that the oracle $U_f$ marks the solutions of the problem by
changing their phase to $-1$. To emphasize this, we adopt the usual convention
which consists of ignoring the last register and considering that $U_f$ has the
following effect:

$$\begin{cases} 
U_f~\ket{{\bf x}} =-\ket{{\bf x}}, \forall {\bf x} \in S\\ 
U_f~\ket{{\bf x}} = \ket{{\bf x}}, \forall {\bf x} \notin S
\end{cases}.$$

\begin{figure}[hbt!]
\input{resources/grover_run}
\caption{First iteration of loop $\mathcal{L}$ in Grover's algorithm: the combs 
  represent the amplitude of each element}
\label{fig:grover_run}
\end{figure}
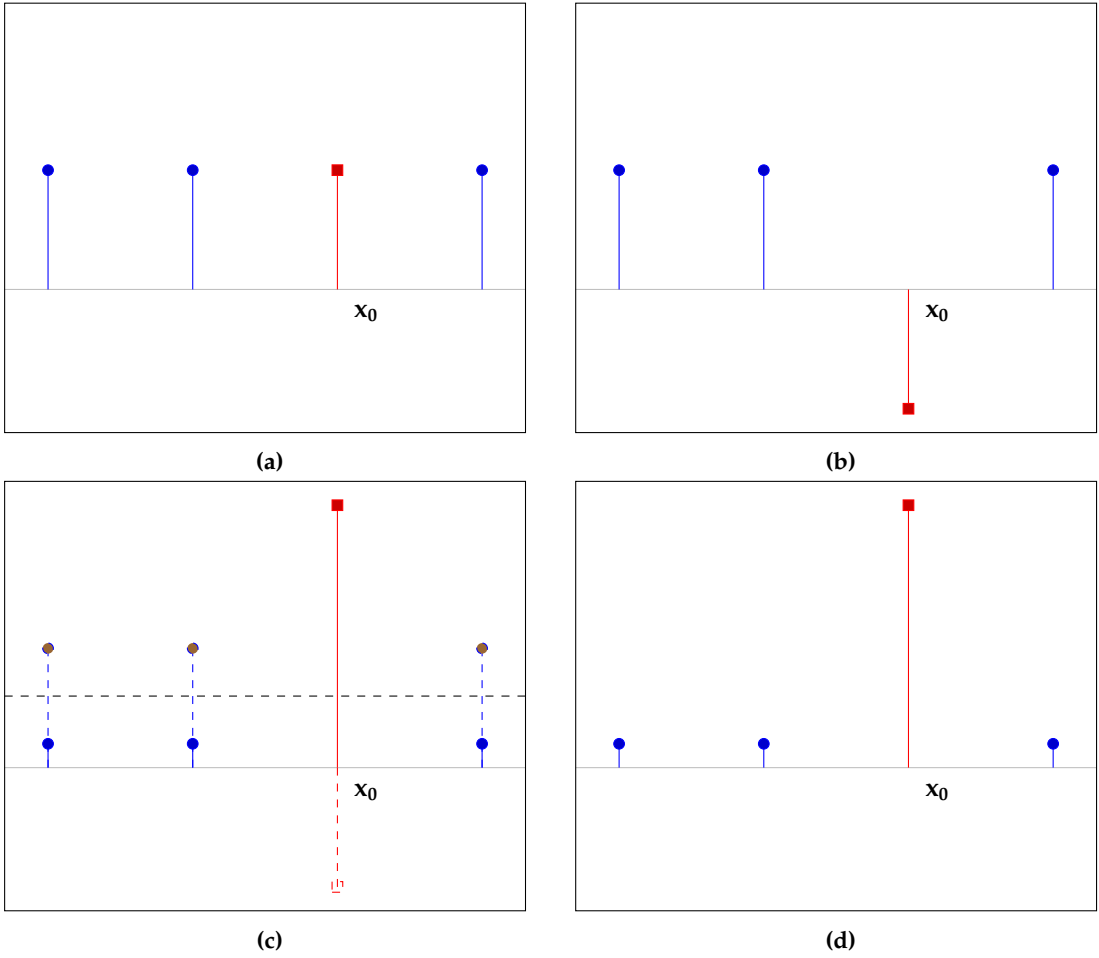  

The diffusion operator $\mathcal{D} = 2(\ket{+}\bra{+})^{\otimes n} - I_{2^n}$
performs the inversion about the mean. Indeed if $\ket{\varphi} =
\sum_{\mathbf{i}=0}^{N-1} \alpha_\mathbf{i}\ket{\mathbf{i}}$ and $\bar{\alpha} =
\frac{1}{N} \sum_{\mathbf{i}=0}^{N-1} \alpha_\mathbf{i}$ denotes the mean value
of the amplitudes of $\ket{\varphi}$,  then  $\mathcal{D}~\ket{\varphi} =
\sum_{\mathbf{i}=0}^{N-1} \alpha'_\mathbf{i}\ket{\mathbf{i}}$ with
$\alpha'_\mathbf{i} - \bar{\alpha} = \bar{\alpha} - \alpha_\mathbf{i}$.

Figure~\ref{fig:grover_run} provides a visualization of the effect of the
beginning of the algorithm on the amplitudes of $\ket{\varphi}$. For readability
purposes, only 4 amplitudes are represented, and only one element is searched 
($S = \{{\bf x_0}\}$), shown with a square instead of a  bullet. The state is
initialized to $\ket{\mathbf{0}}$. The state resulting of applying $H^{\otimes
n}$ is the superposition of all states $\ket{+}^{\otimes n}$
(Figure~\ref{fig:superposition}). Then the oracle $U_f$ flips the searched
element (Figure~\ref{fig:oracle}), and the diffusion operator $\mathcal{D}$
performs the inversion about the mean (Figures~\ref{fig:invertion}
and~\ref{fig:inverted}).

The final measure yields the index of an element from $S$ with high
probability.


  \subsection{Properties of states in Grover's algorithm} 
  \label{sub:properties_of_states_in_the_grover_algorithm}

The evolution of the amplitudes of the state $\ket{\varphi}$ during the
execution of the algorithm is well known~\cite{NC10}. If we denote by $\theta$
the real number such that $\sin(\theta/2)=\sqrt{|S|/N}$, then  after $k$
iterations (\emph{i.e.}, after applying $k$ times the circuit $\mathcal{L}$),
the state is:

\begin{equation}
\label{eq:explicit_states}
\ket{\varphi_k}=\alpha_k \sum_{{\bf x} \in S}\ket{\bf x} + 
  \beta_k \sum_{\bf x \notin S} \ket{\bf x}
\end{equation}

\noindent with $\alpha_k=\dfrac{1}{\sqrt{|S|}}\sin(\dfrac{2k+1}{2}\theta)$ and
$\beta_k=\dfrac{1}{\sqrt{N-|S|}}\cos(\dfrac{2k+1}{2}\theta)$. The sequences
$(\alpha_k)_k$ and $(\beta_k)_k$ are two real sequences respectively increasing
and decreasing  when $k$ varies between 0 and $k_{opt} = \left\lfloor
\frac{\pi}{4} \sqrt{\frac{N}{|S|}}\right\rfloor$.

An alternative representation of the evolution of the states during the
execution of Grover's algorithm is proposed in~\cite{HJN16}. An elementary
algebra calculation (See
Appendix~\ref{sec:explicit_states_for_grover_s_algorithm},
Proposition~\ref{prop:explicit_state_grover}) shows that 
\begin{equation}
\label{eq:explicit_states_bis}
\ket{\varphi_k} = \tilde{\alpha}_k \sum_{\bf x \in S} \ket{\bf x} + 
  \tilde{\beta}_k \ket{+}^{\otimes n} 
\end{equation}

\noindent with $\tilde{\alpha}_k=\alpha_k-\beta_k$ and $\tilde{\beta}_k =
2^{n/2}\beta_k$. The sequences $(\tilde{\alpha}_k)$ and $(\tilde{\beta}_k)$ are
respectively increasing and decreasing on $\llbracket0, k_{opt}\rrbracket$ (see
Appendix~\ref{sec:explicit_states_for_grover_s_algorithm},
Proposition~\ref{prop:grover_a_b_monotonous}).

In particular, if one considers the case of one searched element $\ket{\bf
x_0}$, \emph{i.e.} $S=\{ {\bf x_0} \}$, then 
Equation~(\ref{eq:explicit_states_bis}) becomes
\begin{equation}
\label{eq:classical_explicit_states_one_element}
\ket{\varphi_k} = \tilde{\alpha}_k \ket{\bf x_0} + \tilde{\beta}_k 
  \ket{+}^{\otimes n}.
\end{equation}

\begin{figure}[hbt!]
\begin{center}
\input{resources/HJN_fig2}
\end{center}
\caption{States path (dotted) in relation with the 
variety of separable states ($X$) during Grover's algorithm
execution~\cite[Figure 2]{HJN16}}
\label{fig:grover_sec}
\end{figure}
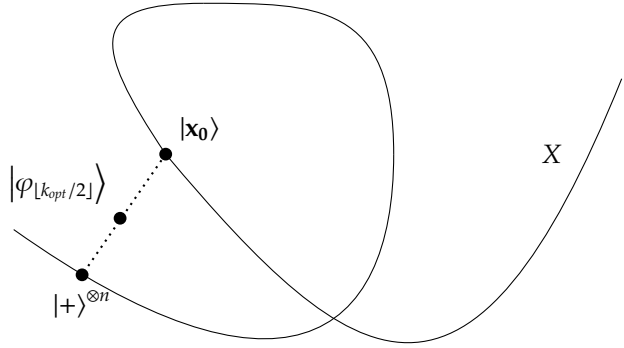

Figure~\ref{fig:grover_sec} provides a pictural interpretation of Equation~(
\ref{eq:classical_explicit_states_one_element}). The ``curve'' $X$ represents
the variety (set defined by algebraic equations) of separable states. In the
picture the evolution of the state $\ket{\varphi_k}$ is seen as a point moving
on a secant line of the set of separable states, starting from the separable
state $\ket{+}^{\otimes n}$ and moving to the separable state $\ket{\bf x_0}$.

In \cite{HJN16}, it is proven that  for states in  superposition $\alpha
\ket{\bf x_0}+\beta\ket{+}^{\otimes n}$ with $\alpha, \beta\in \mathbb{R}_+$,
the GME is maximal when $\alpha=\beta$. Let $\ket{\varphi_{ent}}$ hereafter
denote the state $(\ket{\bf x_0}+\ket{+}^{\otimes n})/K$ normalized with the
factor $K$. Figure~\ref{fig:grover_sec} indicates that the search goes through a
maximally entangled state around the step $k_{opt}/2$ and that the maximally
entangled states generated by Grover's algorithm should be close to that state
$\ket{\varphi_{ent}}$.


\subsection{Quantum Fourier Transform (QFT)}
\label{sub:qft}

The quantum analogous of the Discrete Fourier Transform (DFT) is the Quantum
Fourier Transform (QFT). It acts linearly on quantum registers and is a key
step in Shor's algorithm, permitting to reveal the period of the function
defining the factorization problem \cite{Sho94,NC10}.

In the context of Shor's algorithm, the QFT is used to transform a periodic
state into another one to obtain its period. The periodic state
$\ket{\varphi^{l,r}}$ of $n$ qubits with shift $l$ and period $r$ is defined
by

\begin{equation*}
\label{eq:periodic}
 \ket{\varphi^{l,r}} = \frac{1}{\sqrt{A}} \sum_{i=0}^{A-1}\ket{l+ir}
 \hspace{10mm} \text{with } A = \ceil*{\frac{N-l}{r}} \text{ and } N=2^n,
\end{equation*}
\noindent for $0 \leq l \leq N-1$ and $1 \leq r \leq N-l-1$.
 
For example, for the periodic 4-qubit states, with shift $l=1$ and period $r=5$,
there are $A = \ceil*{\frac{16-1}{5}} = 3$ basis elements, so:

\begin{equation*}
\label{eq:example_4_qubit_periodic}
\ket{\varphi^{1,5}}=\frac{1}{\sqrt{3}}\big(\ket{\bf 1}+\ket{\bf 6}+
\ket{\bf 11}\big)=\frac{1}{\sqrt{3}}\big(\ket{0001}+\ket{0110}+\ket{1011}\big)~.
\end{equation*}

When applied to one of the computational basis states $\ket{k} \in \{\ket{0},
\ket{1}, \dots, \ket{N-1}\}$ (expressed here in decimal notation), the result of 
the QFT can be expressed by
\begin{equation*} 
 QFT~\ket{k} = \frac{1}{\sqrt{N}} \sum_{j=0}^{N-1} \omega^{kj} \ket{j}, 
\end{equation*}
\noindent where $\omega = e^{\frac{2i\pi}{N}}$ is the primitive $N$-th root of
unity. Then, for any $n$-qubit state $\ket{\psi} = \sum_{j=0}^{N-1}x_j \ket{j}$, 
we get
\begin{equation}
\label{eq:qftbase}
  QFT~\ket{\psi} = \sum_{k=0}^{N-1} y_k \ket{k} ~~~~ 
\text{with } y_k = \frac{1}{\sqrt{N}} \sum_{j=0}^{N-1}  x_j \cdot \omega^{k j}.
\end{equation} 

\noindent The corresponding matrix is
\begin{equation*}
  QFT_N = \frac{1}{\sqrt{N}}
  \begin{pmatrix}
  1 & 1 & 1 & 1 & \cdots & 1 \\
  1 & \omega^{1} & \omega^{2} & \omega^{3} & \cdots & \omega^{N-1}\\
  1 & \omega^{2} & \omega^{4} & \omega^{6} & \cdots & \omega^{2(N-1)} \\
  1 & \omega^{3} & \omega^{6} & \omega^{9} & \cdots & \omega^{3(N-1)} \\
  \vdots & \vdots & \vdots & \vdots & \ddots & \vdots \\
  1 & \omega^{N-1} & \omega^{2(N-1)} & \omega^{3(N-1)} & \cdots & \omega^{(N-1)(N-1)} 
  \end{pmatrix}.
\end{equation*}

In the circuit representation, the QFT can be decomposed  into several one-qubit
or two-qubit operators. To obtain this decomposition three different kinds of
gates are used: the Hadamard gate, the SWAP gate and the 
\emph{controlled}-$R_k$ gates, defined by the matrices and circuits
\begin{equation*}
\label{eq:swap}
SWAP = \begin{pmatrix}
     1 & 0 & 0 & 0 \\
     0 & 0 & 1 & 0 \\
     0 & 1 & 0 & 0 \\
     0 & 0 & 0 & 1 \\
   \end{pmatrix} ~~~~~~~~~~~~~~ \Qcircuit @C=1em @R=.7em {
   \lstick{\ket{x} } & \ctrl{1} & \targ     & \ctrl{1}   & \rstick{\ket{y}}\qw\\
   \lstick{\ket{y} } & \targ    & \ctrl{-1} & \targ      & \rstick{\ket{x}}\qw\\
   }
\end{equation*}
and

\begin{equation*}
\label{eq:cRk}
cR_k = \begin{pmatrix}
     1 & 0 & 0 & 0 \\
     0 & 1 & 0 & 0 \\
     0 & 0 & 1 & 0\\
     0 & 0 & 0 & e^{\frac{2i\pi}{2^k}} \\
   \end{pmatrix}  ~~~~~~~~~~~~~~~~ \Qcircuit @C=1em @R=.7em {
   \lstick{\ket{x} }  & \gate{R_k} & \qw\\
   \lstick{\ket{y} }  & \ctrl{-1}      &\qw\\
   }
\end{equation*}
 The complete circuit of the QFT is provided in Figure
\ref{fig:qft_circuit}, where the $n$-qubit SWAP operation consists of swapping
$\ket{x_1}$ with $\ket{x_n}$, $\ket{x_2}$ with $\ket{x_{n-1}}$, and so on.

\begin{figure}[!ht]
   \begin{align*}
   \Qcircuit @C=1em @R=.7em {
   \lstick{\ket{x_1}}   &\gate{H}               &\gate{R_2}       &\gate{R_3}   
      &\ustick{...}\qw    &\gate{R_n}             &\qw              &\qw          
        &\qw                &\qw                    &\qw              &\qw          
          &\qw                &\multigate{6}{SWAP}\\
   \lstick{\ket{x_2}}   &\qw                    &\ctrl{-1}        &\qw          
      &\qw                &\qw                    &\gate{H}         &\gate{R_2}   
        &\ustick{...}\qw    &\gate{R_{n-1}}         &\qw              &\qw     
          &\qw                &\ghost{SWAP}\\  
   \lstick{\ket{x_3}}   &\qw                    &\qw              &\ctrl{-2}  
      &\qw                &\qw                    &\qw              &\ctrl{-1}  
        &\qw                &\qw                    &\ustick{...}\qw  &\qw        
          &\qw                &\ghost{SWAP}\\
   \\
   \lstick{\vdots}      &                       &\vdots           &           
      &\vdots             &                       &\vdots           &         
        &\vdots             &                       &\ustick{...}     &\vdots 
          &                   &\\
   \\
   \lstick{\ket{x_n}}   &\qw                    &\qw              &\qw        
      &\qw                &\ctrl{-6}              &\qw              &\qw        
        &\qw                &\ctrl{-5}              &\ustick{...}\qw  &\gate{H} 
          &\qw                &\ghost{SWAP}\\
   }
   \end{align*}
   \caption{Quantum circuit representation of the Quantum Fourier Transform for 
    a $n$-qubit register}
   \label{fig:qft_circuit}
   \end{figure}
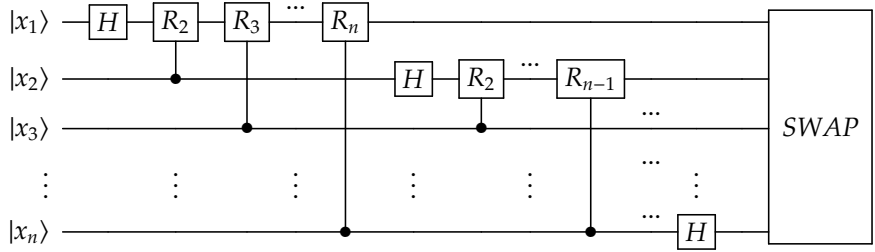

\begin{remark}
One of the reasons that explain the exponential speed-up in Shor's quantum
algorithm, is the complexity of the QFT which is quadratic with respect to the
number of registers. By comparison, classically, the complexity of the Fast
Fourier Transform algorithm that computes the DFT of a vector with $2^n$ entries
is in $\mathcal{O}(n2^n)$.
\end{remark}

  \subsection{Mermin polynomials and Mermin inequalities} 
  \label{sub:mermin_polynomials_and_mermin_inequalities}
  

Entanglement variations during the execution of Grover's algorithm have been
studied either by computing the evolution of the Geometric Measure of
Entanglement~\cite{RBM13,WG03}, or by computing other measures of entanglement
like the concurrence or measures based on invariants~\cite{BOF+16,WG03,HJN16}.
Similarly, for Shor's algorithm and in particular to study the variation of
entanglement within the QFT, numerical computation of the Geometric Measure of
Entanglement was carried in~\cite{SSB05}. Let us also mention~\cite{JH19} where
the evolution of entanglement in Grover and Shor algorithms is studied
qualitatively by considering the classes of entanglement reached during the
execution of the algorithm.

The authors of~\cite{BOF+16} proposed to exhibit the non-local behavior of the
states generated by Grover's algorithm by testing a generalization of Bell's
inequalities known as Mermin's inequalities, based on Mermin
polynomials~\cite{ACG+16,CGP+02}.

\begin{definition}[Mermin polynomials, \cite{ACG+16}]
\label{def:mermin_operator} 
  Let $\left(a_j\right)_{j \geq 1}$ and
  $\left(a_j'\right)_{j \geq 1}$ be  two families of one-qubit observables with 
  eigenvalues in $\{-1,+1\}$. The \emph{Mermin polynomial} $M_n$ is
  inductively defined by:
  \begin{equation} 
    \label{eq:mermin_def}
    \begin{cases}
    M_1=a_1\\
    \forall n \geq 2, \; M_n = \frac{1}{2} M_{n-1}\otimes( a_n + a_n' ) +
    \frac{1}{2} M_{n-1}'\otimes( a_n - a_n' )
    \end{cases}
  \end{equation}
  where, in (\ref{eq:mermin_def}), $M_k'$ is obtained from $M_k$ by 
    interchanging operators with and without the prime symbol.
\end{definition}

\begin{example}
For $n=2$, the Mermin polynomial is
 $M_2=\dfrac{1}{2}(a_1\otimes a_2+a_1\otimes a_2'+a_1'\otimes a_2-a_1'\otimes
 a_2)$. The operator $M_2$ is, up to a factor, the  CHSH operator used to prove
 Bell's Theorem \cite{CHSH69}.
\end{example}

Mermin's inequalities

\begin{equation}
\label{eq:mermin_ineq}
\expval{M_n}^{LR} \leq 1 \hspace{1cm}     \text{and} \hspace{1cm} 
\expval{M_n}^{QM} \leq 2^{\frac{n-1}{2}}
\end{equation}
\noindent respectively formalize that $\ev{M_n}$, the expectation value of $M_n$
is bounded by $1$ under the hypothesis $LR$ of local realism, while it is
bounded by $2^{\frac{n-1}2}$ in quantum mechanics ($QM$).

The violation of the first Mermin's inequality shows non-local behavior which is
only possible under the hypothesis of quantum mechanics and if the quantum state
is entangled. More precisely the maximum violation of Mermin's inequalities
occurs for $GHZ$-like states~\cite{Mer90,CGP+02,ACG+16}, \emph{i.e.} states
equivalent to $\ket{GHZ} = \frac{1}{\sqrt{2}}(\ket{0}^{\otimes n} + \ket{1}^
{\otimes n})$ by local transformations.

One of the advantages of Mermin's inequalities is that they can be
tested by a physical experiment.  Recently the violation of Mermin's
inequalities was tested for $n\leq 5$ qubits on a small quantum
computer~\cite{AL16}.



\section{Method and results} 
\label{sec:method_and_result}

Once two families $(a_j)_{1\leq j\leq n}$ and $(a'_j)_{1\leq j\leq n}$ of
observables are chosen, one can define the \emph{Mermin test function} $f_{M_n}$
by $f_{M_n}(\varphi) = \expval{M_n}{\varphi}$. It comes from
Inequalities~(\ref{eq:mermin_ineq}) that $f_{M_n}(\varphi) >1$ implies that
$\ket{\varphi}$ is entangled. We present in this section two approaches to
choose the parameters  $(a_j)_{1\leq j\leq n}$ and $(a'_j)_{1\leq j\leq n}$ of
$M_n$ to satisfy the previous inequality for some states generated by a quantum
algorithm.

The first approach evaluates each state that the algorithm goes through with the
same function $f_{M_n}$. This approach has the advantage of being light in
calculations ($(a_j)_ {1\leq j\leq n}$ and $(a'_j)_{1\leq j\leq n}$ are computed
only once), but the function $f_{M_n}$ is not a measure of entanglement, since
it is not invariant by local unitary transformations, \emph{i.e.}, we do not
have $f_{M_n}(\varphi) = f_{M_n}(g.\varphi)$ for all transformations $g \in
LU=U_2(\mathbb{C}) \times\dots\times U_2(\mathbb{C})$ and all quantum states
$\ket{\varphi}$. Here, for $g=(g_1,\dots,g_n)$ and $G = g_1 \otimes \ldots
\otimes g_n$, $\ket{g.\varphi} = G~\ket{\varphi}$.

The second approach is to apply a different function $f_{M_n}$ to each  state
$\ket{ \varphi}$ traversed by the algorithm, by finding values for $(a_j)$ and
$(a_j')$ such as $f_{M_n}(\varphi) >1$ for some of these states. This approach
was for example used in~\cite{BOF+16} and we use it in the
Section~\ref{sub:method_qft} to define a quantity $\mu(\varphi)$, invariant
under the group $LU$ of local unitary transformations, that could be considered
as a measure of entanglement (see Proposition~\ref{prop:measure_of_ent}).

\subsection{Grover's algorithm properties}

Hereafter we simplify the calculations by taking $S = \{{\bf x_0}\}$, \emph{i.e.
}, by considering that Grover's algorithm is only searching for a single element 
$\bf x_0$.
We want to show two properties: 
\begin{enumerate}[font=\bfseries]
  \item Grover's algorithm exhibits non-local behavior,
  \item the values of a well-chosen Mermin test function for 
    the successive states $\ket{\varphi_k}$ in Grover's algorithm increase 
    and then decrease, reaching their maximum at an integer $k_{max}$ in 
    $\{\floor{k_{opt}/{2}},\ceil{k_{opt}/{2}}\}$ (\emph{i.e.}, the chosen
    Mermin test function behaves like a measure of entanglement).
\end{enumerate}

The next section details the method we have followed to find a good Mermin
polynomial to establish these properties. 

  \subsubsection{Method} 
  \label{sub:method_grover}

The definition of Mermin polynomials provides degrees of freedom in the choice
of $(a_j)_{j \geq 1}$ and $(a_j')_{j \geq 1}$ (an infinite number of
parameters). We reduce that choice by imposing that the two sequences $(a_j)_{j
\geq 1}$ and $(a_j')_{j \geq 1}$ are constant, \emph{i.e.} $\forall j, a_j = a
\text{ and } a'_j = a'$. This restriction makes calculations lighter, and it
will be sufficient to achieve our objectives.

Let us denote by $a$ and $a'$ the two one-qubit observables that will be used to
write our Mermin polynomial. We have $a=\alpha X+\beta Y+\gamma Z$ and
$a'=\alpha'X+\beta'Y+\gamma'Z$ with the constraints
$|\alpha|^2+|\beta|^2+|\gamma|^2=1$ and $|\alpha'|^2+|\beta'|^2+|\gamma'|^2=1$
where $X=\begin{pmatrix} 0 & 1\\ 1 & 0
\end{pmatrix}$, $Y=\begin{pmatrix} 0 & -i\\ i & 0 \end{pmatrix}$ and
$Z=\begin{pmatrix} 1 & 0\\ 0 & -1 \end{pmatrix}$ denote the usual Pauli
matrices.

The degrees of freedom are the $6$ complex numbers $\alpha$, $\beta$,
$\delta$, $\alpha'$, $\beta'$ and $\delta'$ with the two normalization
constraints. Let $A = (\alpha, \beta, \delta, \alpha', \beta', \delta')$ be the
six-tuple of these variables.

In order to satisfy Property {\bf 2}, we search for a six-tuple of parameters
$A$ such that $f_{M_n}$ reaches its maximum for the state $\varphi_{k_{opt}/2}$.
We also would like this choice of $A$ to be independent of the states generated
by the algorithm. According to the geometric interpretation presented in
Section~\ref{sub:properties_of_states_in_the_grover_algorithm}, the state
$\varphi_{k_{opt}/2}$ should tend to the state $\ket{\varphi_{ent}} =
\frac{1}{K}(\ket{\bf x_0} + \ket{+}^{\otimes n})$ when $n$ tends to infinity 
(the approximation improves as $n$ increases). Moreover the state
$\ket{\varphi_{ent}}$ is a rank two tensor with an overlap between the states
$\ket{{\bf x_0}}$ and $\ket{+}^{\otimes n}$ which tends to $0$ as $n$ increases,
\emph{i.e.}, we expect the state $\ket{\varphi_{ent}}$ to behave like a
$GHZ$-like state when $n$ is big. This point is important because $GHZ$-like
states are the ones that maximize the violation of classical inequalities by
Mermin polynomials~ \cite{Mer90,CGP+02,ACG+16}. Therefore by choosing a tuple of
parameters $A$ maximizing $f_{M_n}(\varphi_{ent})$ one expects to satisfy
Properties {\bf 1.} and {\bf 2.}.

We use a random walk in $\mathbb{R}^6$ to maximize $f_{M_n}({\varphi_{ent}})$.
We operate the walk for a fixed number of steps, starting form an arbitrary
point. At each step, we chose a random direction, and move toward it to a new
point. If the value of $f_{M_n}({\varphi_{ent}})$ at that new point is higher
than at the previous one, then that point is the start point for the next step,
otherwise a new point is chosen.

Once the proper coefficient for $M_n$ found, we compute the values of each $f_
{M_n}({\varphi_{k}})$ for $k$ in $\llbracket0,k_{opt}\rrbracket$ to validate 
Properties {\bf 1.} and {\bf 2.}.

\begin{example}
When searching the state $\ket{0000}$, the highest value of $f_{M_4}({\varphi_
{ent}})$ obtained by this random walk was for $A = (-0.7, -0.3, -0.7, -0.5, 0.7,
-0.5)$. Then, $A$ is used to compute $M_4$, and then $f_{M_4}({\varphi_k}),
\forall k  \in  \llbracket0, k_{opt}\rrbracket$. 
\end{example}

\begin{remark}
\label{rem:BOFcalc}
Some comments should be done at this point to compare our approach with the work
of \cite{BOF+16}. First in \cite{BOF+16} all calculations are done using the
density matrices formalism instead of the vector/tensor approach we use here,
which is sufficient for computation involving pure states. Moreover in
\cite{BOF+16} the optimization is done at each step of the algorithm with
respect to the state computed by the algorithm, while we compute the parameters
only once with respect to a targeted state $\ket{\varphi_{ent}}$. Finally, as
mentioned at the beginning of Section~\ref{sub:method_grover}, we also restrict
ourselves to two operators $a$ and $a'$ and thus all optimizations are 
performed on six parameters instead of $6n$. This allows us to perform the
calculation for a larger number of qubits (up to $12$).
\end{remark}

  
  \subsubsection{Results} 
  \label{sub:results}

Thanks to our implementation of this method in SageMath, described in
Section~\ref{sec:implementation}, we obtain the values depicted in
Figure~\ref{fig:mermin_experiment}, for $n$ from 4 up to 12 qubits. The searched
element ${\bf x_0}$ is always the first element $\ket{\bf 0}$ of the canonical
basis, but other searched elements give similar results, by symmetry of the
problem.

\begin{figure}[htb!]
\input{resources/mermin_for_grover}
\caption{Violation of Mermin's inequalities during Grover's algorithm execution 
for $4\leq n\leq 12$ qubits}
\label{fig:mermin_experiment}
\end{figure}
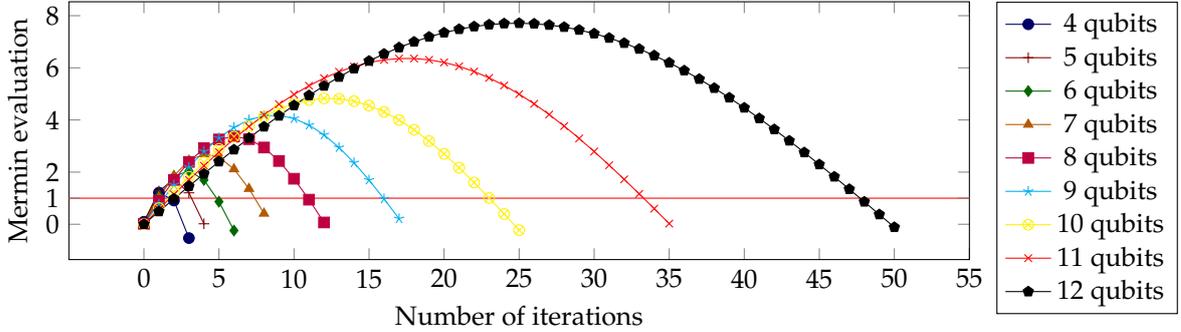

The lower bound for the number $n$ of qubits is set to 4 because for $n \leq 3$
the algorithm has no time to show any advantage, is not very reliable and
doesn't exhibit non locality. The upper bound is set to 12 because of
technological limitations: computations for 13 qubits or more become too
expensive.

We see that the two expected properties hold for all values of $n$: the
classical limit is violated and the Mermin evaluation increases  up to the
middle of the executions, and then decreases (the maximal values are given in
Figure~\ref{fig:k_opt_table}).

The curve for $n = 12$ in Figure~\ref{fig:mermin_experiment} should be compared
to the curve of Figure~1 of~\cite{RBM13} where the evolution of the GME of the
states generated by Grover's algorithm is given for $n=12$ qubits. In our
setting it is not a surprise that both curves are similar because in all of our
calculations the function $f_{M_n}$ is defined by the set of parameters that
maximizes its value for $\ket{\varphi_{ent}}$.

\begin{figure}[htb!]
$$\begin{array}{|l|c|c|c|c|c|c|c|c|c|}
\hline
n                          & 4  & 5  & 6  & 7  & 8  & 9  & 10 & 11 & 12 \\
\hline
\floor{k_{opt}/2}          & 1  & 1  & 2  & 3  & 5  & 8  & 12 & 17 & 24 \\
\hline
k_{max}                    & 1  & 2  & 3  & 4  & 6  & 9  & 12 & 18 & 25 \\
\hline
f_{M_n}(\varphi_{k_{max}}) &1.21&1.72&2.05&2.69&3.37&4.17&4.83&6.36&7.71\\
\hline
\end{array}$$
\caption{Maximums of $f_{M_n}(\varphi_k)$ for $4\leq n\leq 12$ qubits}
\label{fig:k_opt_table}
\end{figure}

\begin{remark}
In~\cite{BOF+16} similar curves (Figure~3) were obtained for  $n \in \{2,4,6,8
\}$ qubits showing the increasing-decreasing behavior, but the violation of
Mermin's inequalities - the non-locality - was not established for $n=6$ and
$n=8$, whereas it is obtained in our calculation.  Recall from Remark
\ref{rem:BOFcalc} that the calculation of \cite{BOF+16} is not exactly the same
as the one performed in this paper. The curves of \cite{BOF+16} are obtained by
maximizing $f_{M_n}(\varphi_k)$ at each step of the algorithm with a larger
number of parameters. Therefore as we obtain violation of Mermin's inequalities
in a restricted calculation, the authors of \cite{BOF+16} should also have
observed it. We suspect errors in the implementation of the calculation of
Equations (19) of \cite{BOF+16} as we have redone this calculation for $n=6$
based on Equations (18) and (20) of \cite{BOF+16}, and we have obtained
violation of Mermin inequalities shown in Figure~\ref{fig:grover_bof}.
\end{remark}

\begin{figure}[htb!]
\input{resources/Grover_BOF}
\caption{Violation of Mermin inequalities during Grover's algorithm execution
for 6 qubits using \cite{BOF+16} method}
\label{fig:grover_bof}
\end{figure}
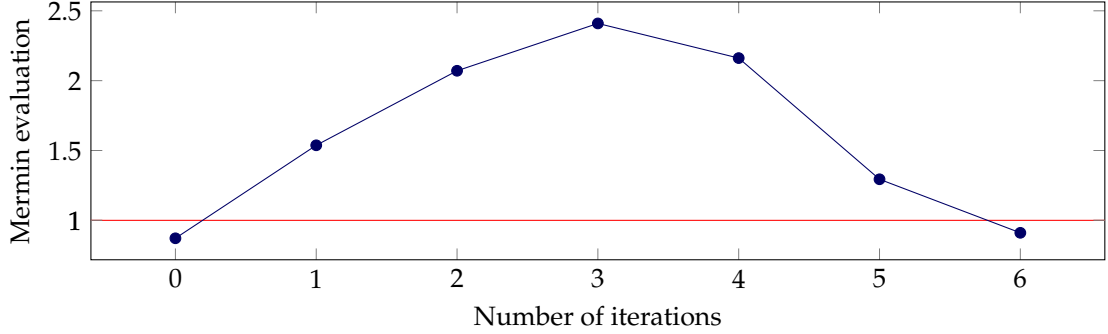

Figure~\ref{fig:mermin_comparison} provides another argument explaining why we
expected violation of Mermin's inequalities in Grover's algorithm when $n$
increases. From the geometric description of the algorithm (subsection
\ref{sub:properties_of_states_in_the_grover_algorithm}) one knows that the
quantum state $\ket{\varphi_{\ceil{k_{opt/2}}}}$ should be close to
$\ket{\varphi_{ent}}$ and thus behave like it with respect to Mermin's
polynomial. Despite the fact that $f_{M_n}(\varphi_{ent})$ does not
reach the theoretical upper bound that is obtained for states LOCC equivalent to
$\ket{GHZ_n}$, one sees that the difference between $f_{M_n}(\varphi_{ent})$
and the classical bound $1$ increases as a function of $n$.

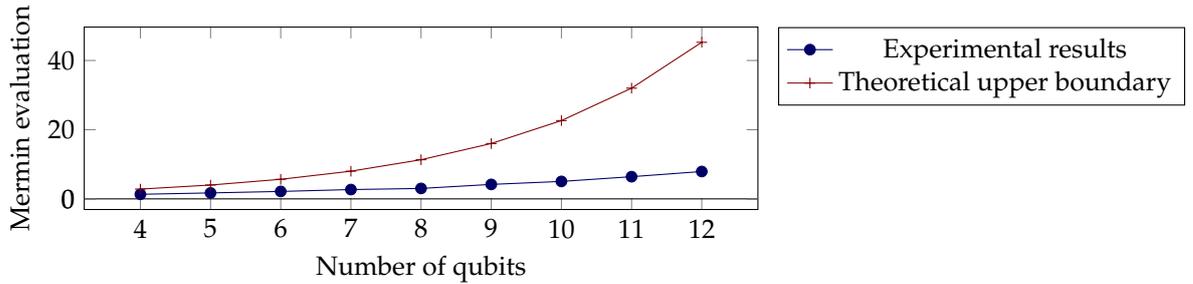
\begin{figure}[htb!]
\input{resources/max_exp-theo}
\caption{Comparison between experimental results and theoretical Mermin
boundary. The curve with points as dots corresponds to the evaluation of 
$f_{M_n}({\varphi_{ent}})$ and the curve with points as crosses corresponds to
the theoretical upper bound for the violation of Mermin's inequality defined by
$M_n$}
\label{fig:mermin_comparison}
\end{figure}

\subsection{Quantum Fourier Transform}

To exhibit non-locality behavior of states generated at each step of the Quantum
Fourier Transform we restrict ourselves to periodic four-qubit states for
the following reasons:
\begin{enumerate}
  \item as explained in subsection~\ref{sub:qft}, the QFT in Shor's algorithm 
    is applied to periodic states~\cite{NC10};
  \item as we will see in subsection~\ref{sub:resultqft} the four-qubit case
  is sufficient to obtain violation of Mermin's inequalities;
  \item we want to compare the present approach with a recent study
    of entanglement in Shor's algorithm in the four-qubit case, proposed by two 
    of the authors of the present paper~\cite{JH19}.
\end{enumerate}

\subsubsection{Method}
\label{sub:method_qft}

When we apply the QFT to periodic states we have no \emph{a priori} geometric
information about the type of states that will be generated. In fact it depends
on two initial parameters that define the periodic state $\ket{\varphi^{l,r}}$:
its shift $l$ and its period $r$. Therefore there are no reasons for restricting
the choice of parameters in the calculation of $f_{M_n}(\varphi^{l,r})$. For the
four-qubit case this implies that our optimization will be carried over the $24$
parameters defining $M_4$, hereafter denoted $\alpha_1$, \ldots, $\alpha_{24}$.

For $k\geq 0$, let $\ket{\varphi^{l,r}}_k$ denote the state reached after the
first $k$ gates in the QFT (Figure~\ref{fig:qft_circuit_4_qubits}) initialized
with the periodic state $\ket{\varphi^{l,r}}$ with shift $l$ and period $r$.

\ifdeftostring{\version}{arxiv}
  {
\begin{figure}[!ht]
  \begin{align*}
  \Qcircuit @C=1em @R=.7em {
  \lstick{ }
    &\gate{H}
        &\gate{R_2}
            &\gate{R_3}
                &\gate{R_4}
                    &\qw
                        &\qw
                            &\qw
                                &\qw
                                    &\qw
                                        &\qw
                                            &\multigate{3}{SWAP}
                                                &\qw \\
  \lstick{ }                     
    &\qw                          
        &\ctrl{-1}                      
            &\qw                            
                &\qw                            
                    &\gate{H}                
                        &\gate{R_2}              
                            &\gate{R_3}              
                                &\qw                     
                                    &\qw                     
                                        &\qw                      
                                            &\ghost{SWAP} 
                                                &\qw \\  
  \lstick{ }                     
    &\qw                          
        &\qw                            
            &\ctrl{-2}                      
                &\qw                            
                    &\qw                     
                        &\ctrl{-1}               
                            &\qw                     
                                &\gate{H}                
                                    &\gate{R_2}              
                                        &\qw                      
                                            &\ghost{SWAP} 
                                                &\qw \\
  \lstick{ }                     
    &\qw                          
        &\qw                            
            &\qw                            
                &\ctrl{-3}                      
                    &\qw                     
                        &\qw                     
                            &\ctrl{-2}               
                                &\qw                     
                                    &\ctrl{-1}               
                                        &\gate{H}                 
                                            &\ghost{SWAP} 
                                                &\qw \\
  \dstick{\quad\ket{\varphi_{0}^{l,r}}}
      &\dstick{\qquad\;\;\ket{\varphi_{1}^{l,r}}}
          &\dstick{\qquad\quad\ket{\varphi_{2}^{l,r}}}
              &\dstick{\qquad\quad\ket{\varphi_{3}^{l,r}}}
                  &\dstick{\qquad\quad\ket{\varphi_{4}^{l,r}}}
                      &\dstick{\qquad\quad\ket{\varphi_{5}^{l,r}}}
                          &\dstick{\qquad\quad\ket{\varphi_{6}^{l,r}}}
                              &\dstick{\qquad\quad\ket{\varphi_{7}^{l,r}}}
                                  &\dstick{\qquad\quad\ket{\varphi_{8}^{l,r}}}
                                      &\dstick{\qquad\quad\ket{\varphi_{9}^{l,r}}}
                                          &\dstick{\qquad\quad\ket{\varphi_{10}^{l,r}}}
                                              &\dstick{\qquad\qquad\quad\ket{\varphi_{11}^{l,r}}}
  }    
  \end{align*}
  \caption{Quantum circuit representation of the Quantum Fourier Transform for 
    a 4-qubit register}
  \label{fig:qft_circuit_4_qubits}
\end{figure}
  }
\ifdeftostring{\version}{qip}
  {
\begin{figure}[!ht]
  \begin{align*}
  \Qcircuit @C=1em @R=.7em {
  \lstick{ } \barrier[-1.2em]{3} 
    &\gate{H} \barrier[-1.3em]{3} 
        &\gate{R_2} \barrier[-1.3em]{3} 
            &\gate{R_3} \barrier[-1.3em]{3} 
                &\gate{R_4} \barrier[-1.2em]{3} 
                    &\qw \barrier[-1.3em]{3} 
                        &\qw \barrier[-1.3em]{3} 
                            &\qw \barrier[-1.2em]{3} 
                                &\qw \barrier[-1.3em]{3} 
                                    &\qw \barrier[-1.2em]{3} 
                                        &\qw \barrier[-2.6em]{3} 
                                            &\multigate{3}{SWAP} \barrier[-.4em]{3} 
                                                &\qw \\
  \lstick{ }                     
    &\qw                          
        &\ctrl{-1}                      
            &\qw                            
                &\qw                            
                    &\gate{H}                
                        &\gate{R_2}              
                            &\gate{R_3}              
                                &\qw                     
                                    &\qw                     
                                        &\qw                      
                                            &\ghost{SWAP} 
                                                &\qw \\  
  \lstick{ }                     
    &\qw                          
        &\qw                            
            &\ctrl{-2}                      
                &\qw                            
                    &\qw                     
                        &\ctrl{-1}               
                            &\qw                     
                                &\gate{H}                
                                    &\gate{R_2}              
                                        &\qw                      
                                            &\ghost{SWAP} 
                                                &\qw \\
  \lstick{ }                     
    &\qw                          
        &\qw                            
            &\qw                            
                &\ctrl{-3}                      
                    &\qw                     
                        &\qw                     
                            &\ctrl{-2}               
                                &\qw                     
                                    &\ctrl{-1}               
                                        &\gate{H}                 
                                            &\ghost{SWAP} 
                                                &\qw \\
  \dstick{\quad\ket{\varphi_{0}^{l,r}}}
      &\dstick{\qquad\;\;\ket{\varphi_{1}^{l,r}}}
          &\dstick{\qquad\quad\ket{\varphi_{2}^{l,r}}}
              &\dstick{\qquad\quad\ket{\varphi_{3}^{l,r}}}
                  &\dstick{\qquad\quad\ket{\varphi_{4}^{l,r}}}
                      &\dstick{\qquad\quad\ket{\varphi_{5}^{l,r}}}
                          &\dstick{\qquad\quad\ket{\varphi_{6}^{l,r}}}
                              &\dstick{\qquad\quad\ket{\varphi_{7}^{l,r}}}
                                  &\dstick{\qquad\quad\ket{\varphi_{8}^{l,r}}}
                                      &\dstick{\qquad\quad\ket{\varphi_{9}^{l,r}}}
                                          &\dstick{\qquad\quad\ket{\varphi_{10}^{l,r}}}
                                              &\dstick{\qquad\qquad\quad\ket{\varphi_{11}^{l,r}}}
  }    
  \end{align*}
  \caption{Quantum circuit representation of the Quantum Fourier Transform for 
    a 4-qubit register}
  \label{fig:qft_circuit_4_qubits}
\end{figure}
  }

We are interested by the evolution of the function $q$ defined for $k \geq 0$ by
\begin{equation*}
q(k)=\text{Max}_{\alpha_1,\dots,\alpha_{24}}f_{M_4}\left(\varphi^{l,r}_k\right).
\end{equation*}

In~\cite{JH19} two of the authors of the present paper have studied the
evolution of entanglement for periodic four-qubit states through QFT by
computing the absolute value of an algebraic invariant called the Cayley
hyperdeterminant and denoted by $\Delta_{2222}$. This polynomial of degree $24$
in $16$ variables is a well-known invariant in quantum information
theory~\cite{MW02,OS06} and was introduced as a possible measure of entanglement
for four-qubit case in~\cite{GW14}. We provide the definition of $\Delta_{2222}$
in Appendix~\ref{appendix:cayley_hyperdeterminant_D2222}.

Surprisingly, the two approaches, which are of different natures -- algebraic
definition for the hyperdeterminant and an operator-based construction for
Mermin evaluation -- would sometimes present similar behavior (see
Figure~\ref{fig:qft_hyperdet_comparison}).

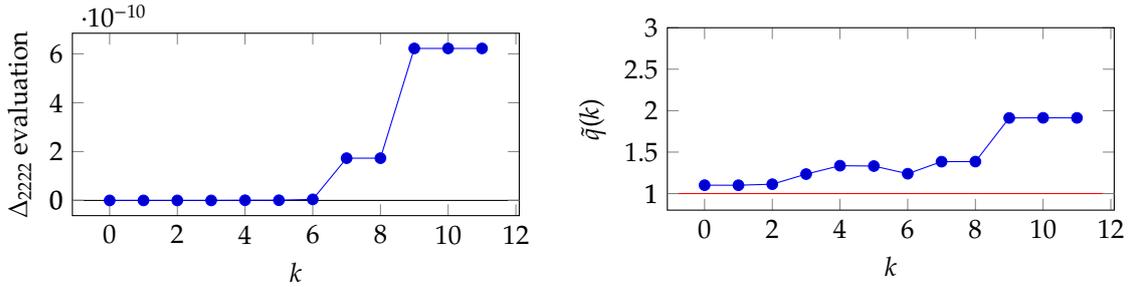
\begin{figure}[!ht]
  \input{resources/qft_hyperdet_comparison}
  \caption{Comparison of entanglement evaluation through the QFT for periodic 
    state $(l,r)=(9,1)$ using the measures given by the hyperdeterminant and the
   Mermin evaluation}
  \label{fig:qft_hyperdet_comparison}
\end{figure}

In \cite{JH19} it was observed that  the evolution of entanglement for
four-qubit periodic states through  QFT shows three different behaviors with
respect to $\Delta_{2222}$.
\begin{itemize}
  \item Case 1. The polynomial $\Delta_{2222}$ is nonzero when evaluated on 
    $\ket{\varphi^{l,r}}$ and does not vanish during the transformation. In 
    terms of four-qubit classification~\cite{VDDV02} it means that the
    transformed states remain in the so-called $G_{abcd}$ class. This happens
    for $(l,r) \in \{(1,3),(2,3)\}$.
  \item Case 2. The polynomial $\Delta_{2222}$ is zero for the periodic state 
    $\ket{\varphi^{l,r}}$ and is nonzero during the QFT. This happens for $(l,r)
    \in \{(0,3),(0,5),(2,1),(3,1),(3,3),(4,1),(4,3),(5,1),(5,3),(6,1),(6,3),$
    $(7,1),(9,1)(10,1),(11,1),(12,1)\}$.
  \item Case 3. The polynomial $\Delta_{2222}$ is zero for the periodic state 
    $\ket{\varphi^{l,r}}$ and it remains equal to zero all along the
    QFT for all the other $(l,r)$ configurations (in $\llbracket 0, N-1
    \rrbracket \times \llbracket 1, N-r \rrbracket$).
\end{itemize}

Before presenting the results let us point out that now our calculation
can be considered as a measure of entanglement, because the calculated quantity 
is invariant under local unitary transformations, \emph{i.e.} under the group 
$LU = U_2(\mathbb{C})^n$.

\begin{proposition}
\label{prop:measure_of_ent}
Let $\ket{\varphi}\in {(\mathbb{C}^2)}^{\otimes n}$ be a $n$-qubit state and $(a_i
)$ and $(a_i')$ be families of one-qubit observables that define a Mermin
polynomial $M_n$ according to Definition~\ref{def:mermin_operator}. Let
\begin{equation}
  \mu(\varphi)=\text{Max}_{a_i,a_{i'}} \expval{M_n}{\varphi}.
\end{equation}
Then $\mu(\varphi)$ is LU-invariant.
\end{proposition}

\proof First one recalls that a one-qubit observable $A$ such that
$Sp(A)=\{-1,1\}$ can always be written as $A=\alpha X+\beta Y+\gamma Z$ with
$\alpha,\beta,\gamma\in \mathbb{R}$ and $\alpha^2+\beta^2+\gamma^2=1$. For the
action  $g.A= g^\dagger Ag$ on $A$ by conjugation with a unitary matrix $g\in
U_2(\mathbb{C})$, one has  $g.A =\tilde{A}=\tilde{\alpha} X+\tilde{\beta}
Y+\tilde{\gamma} Z$ with $\tilde{\alpha},\tilde{\beta},\tilde{\gamma}$ reals
such that $\tilde{\alpha}^2+\tilde{\beta}^2+\tilde{\gamma}^2=1$. Indeed
$\tilde{A}$ is also a one-qubit observable such that $Sp(\tilde{A})=\{-1,1\}$.

Let us denote by $\lambda = (\alpha_1,\beta_1,\gamma_1,\alpha'_1,\beta'_1,
\gamma'_1,\dots,\alpha_n,\beta_n,\gamma_n,\alpha'_n,\beta'_n,\gamma'_n)$ a tuple 
of $6n$ parameters that define a Mermin polynomial $M_n(\lambda)$. Then 
\begin{equation*}
  \mu(\varphi)=\text{Max}_{\lambda \in \mathbb{R}^{6n},
    \alpha_i^2+\beta_i^2+\gamma_i^2=1,{\alpha'}_i^2+{\beta'}_i^2+
    {\gamma'}_i^2=1}\expval{M_n(\lambda)}{\varphi}
\end{equation*} 
exists, because it is the maximum of a degree $n$ polynomial in (at most) $6n$
variables under the constraints $\alpha_i^2+\beta_i^2+\gamma_i^2=1$ and ${\alpha
'}_i^2+{\beta'}_i^2+{\gamma'}_i^2=1$. Let us denote by $\lambda'$ a tuple of
parameters that maximizes $\expval{M_n(\lambda)}{\varphi}$, \emph{i.e.},
\begin{equation*}
  \mu(\varphi) = \expval{M_n(\lambda')}{\varphi}.
\end{equation*}

Let $\ket{\psi}$ be a $n$-qubit state $LU$-equivalent to $\ket{\varphi}$. Thus
there exists $g=(g_1,\dots,g_n)\in LU$ such that $\ket{\psi} = \ket{g.\varphi} =
G~\ket{\varphi}$ with $G = g_1 \otimes \ldots \otimes g_n$. Then $\expval{
M_n(\lambda')}{\varphi} = \left(\bra{\varphi}~G^\dagger\right)~G~M_n
(\lambda')~G^\dagger~\left(G~\ket{\varphi}\right) = 
\expval{M_n(\lambda'')}{\psi}$ for some tuple of parameters $\lambda''$. 
Therefore
\begin{equation*}
  \mu(\varphi) \leq \mu(\psi).
\end{equation*}
But $\ket{\varphi}=G^{\dagger}\ket{\psi}$ also holds, so a similar reasoning
provides the inequality $\mu(\varphi) \geq
\mu(\psi)$ and thus the equality. \qed 

In the next section we plot and analyze different curves of the experimental
approximation $\tilde{q}$ of $q$ in the four-qubit case for different choices of 
$(l,r)$.

\subsubsection{Results}
\label{sub:resultqft}

Curves of the experimental approximation $\tilde{q}(k)$ of $q(k)$ are shown on
Figures~\ref{fig:qft_res_hyperdet_not_null},
~\ref{fig:qft_res_hyperdet_not_null_after} and~\ref{fig:qft_res_hyperdet_null},
for $k \in \llbracket 0,11\rrbracket$ and for different choices of shift $l$ and
period $r$, respectively in Cases 1, 2 and 3.

\begin{figure}[!th]
  \input{resources/qft_res_hyperdet_not_null}
  \caption{Evolution of the maximal values of Mermin operators in the QFT steps. 
    Examples of input $\ket{\varphi^{(l,r)}}$ in Case 1 from~\cite{JH19}}
  \label{fig:qft_res_hyperdet_not_null}
\end{figure}

Let us start with general comments.
\begin{itemize}
  \item All examples in Figures~\ref{fig:qft_res_hyperdet_not_null},
    \ref{fig:qft_res_hyperdet_not_null_after} and 
    \ref{fig:qft_res_hyperdet_null} present violations of Mermin's inequality,
    and the maximal violation evolves during the algorithm.
  \item The intervals $\llbracket 0,1\rrbracket, \llbracket 4,5\rrbracket,
    \llbracket 7,8\rrbracket$  and $\llbracket 9,11\rrbracket$ for $k$
    correspond to gates (Hadamard, SWAP) of the QFT that do not modify
    entanglement.That explains why the function is constant on those intervals, 
    as it was already the case for the curves $k \mapsto |\Delta_{2222}(\varphi^
    {l,r}_k)|$ in~\cite{JH19}. 
  \item States corresponding to Cases 1 and 2 of~\cite{JH19} violate the
    classical bound during the execution of the QFT. Only some states
    corresponding to Case 3 produce constant curves with some of them equal to
    the classical bound (not drawn). It is for instance the case for $(l,r)=(2,4
    )$ which is a separable state that remains separable during the algorithm.
\end{itemize}

\begin{figure}[!th]
  \input{resources/qft_res_hyperdet_not_null_after}
  \caption{Evolution of the maximal values of Mermin operators in the QFT steps. 
    Examples of input $\ket{\varphi^{(l,r)}}$ in Case 2 from~\cite{JH19}}
  \label{fig:qft_res_hyperdet_not_null_after}
\end{figure}

\begin{figure}[!th]
  \input{resources/qft_res_hyperdet_null}
  \caption{Evolution of the maximal values of Mermin operators in the QFT steps.
    Examples of input $\ket{\varphi^{(l,r)}}$ in Case 3 from~\cite{JH19}}
  \label{fig:qft_res_hyperdet_null}
\end{figure}

It would be interesting to propose a finer analysis of the evolution of these
curves with respect to the change of entanglement classes induced by the
algorithm. For instance, if one considers the periodic states
$\ket{\varphi^{l,r}}$ for $(l,r)=(2,2)$ and $(l,r)=(0,11)$
(Figures~\ref{fig:qft-l2r2} and~\ref{fig:qft-l0r11}), it is shown in~\cite{JH19}
that these two states are SLOCC equivalent (\emph{i.e.} can be inter-converted
by a reversible local operation), but their evolution during the QFT is quite
different. The value of $\tilde{q}(k)$ fluctuates around $1.10$ for
$(l,r)=(2,2)$ but for $(l,r)=(0,11)$ that value of $\tilde{q}(k)$ is in the
interval $[1.65,2.18 ]$. It was also shown in~\cite{JH19} that the states
$\ket{\varphi^{2,2}_{11}}$ and $\ket{\varphi^{0,11}_{11}}$ are not SLOCC
equivalent.

Similarly the cases $(l,r)=(0,15)$ and $(1,1)$
(Figure~\ref{fig:qft_res_hyperdet_null} bottom) correspond to two states SLOCC
equivalent to $\ket{GHZ_4}$ at the beginning of the algorithm. It is clear for
$(l,r)=(0,15)$ because $\ket{\varphi^{0,15}} = \ket{GHZ_4}$ and $\tilde{q}(k)$
reaches the maximal possible value at the beginning of the algorithm. The
maximal violation of Mermin inequality for four qubits is $2\sqrt{2}\approx
2.81$ ($2^{\frac{n-1}{2}}$ for $n=4$), but this value is nowhere to be
approached for $(l,r)=(1,1)$ where the value of $\tilde{q}(k)$ is close to $1$
at all steps of the run. In fact the state

\begin{equation}
 \ket{\varphi^{1,1}}=\ket{+{+{+{+}}}} - \dfrac{1}{\sqrt{15}}\ket{0000}
\end{equation}

\noindent is a state on the secant line joining $\ket{+{+{+{+}}}}$ and 
$\ket{0000}$, as described in
subsection~\ref{sub:properties_of_states_in_the_grover_algorithm}. This state is
indeed SLOCC equivalent to $\ket{GHZ_4}$ but it is closer to a separable state
if one considers the GME. That could explain the difference of observed values
in Figure~\ref{fig:qft_res_hyperdet_null}.



\section{Implementation} 
\label{sec:implementation}

This section explains the code developed for this article and relates it to the
notations from Section~\ref{sec:background}. This code can be found at
\quantcert. It uses the open-source mathematics software system
SageMath\footnote{\url{http://www.sagemath.org}} based on Python. The code is a
module named \py{mermin_eval}, and usage examples can be found in the GitHub
repository. Note that all the results of this article have been double checked,
by first being obtained on Maple\footnote{ \url{https://www.maplesoft.com/}} and
then only being generalized on SageMath.

The code is provided and presented for several reasons: so the readers can see
how we obtained the results presented in Section~\ref{sub:results}, and they can
reproduce our computations by running the code. But the code can also be
extended to other evaluation methods of Grover algorithm, or adapted to other
quantum algorithms, since it is structured in several well-documented functions.

This section is divided in two parts: we first explain the code used for
Grover's algorithm in Section~\ref{sub:grover_s_algorithm_implementation}, and
then the code used for the Quantum Fourier Transform in
Section~\ref{sub:quantum_fourier_transform_implementation}.

\subsection{Grover's algorithm implementation} 
\label{sub:grover_s_algorithm_implementation}

For Grover's algorithm, the main function \py{grover} is reproduced in
Listing~\ref{lst:grover_func}. The parameter \py{target_state_vector} is the
searched state $\ket{\varphi_0}$. The function first executes an implementation
\py{grover_run} of the Grover algorithm, detailed in
Section~\ref{sub:grover_execution}, and stores in the list \py{end_loop_states}
the states after each iteration of the loop $\mathcal{L}$. Then but
independently, a call to the function \py{grover_optimize}
(Section~\ref{sub:grover_optimization}) optimizes Mermin operator. The result is
stored in the matrix \py{M_opt}. Finally both these results are used to evaluate
entanglement after each iteration of $\mathcal{L}$ with a call to the function
\py{grover_evaluate} (Section~\ref{sub:grover_evaluation}), also responsible of
printing the evaluations at each step.

\begin{listing}
\begin{minted}{python}
def grover(target_state_vector):
  end_loop_states = grover_run(target_state_vector)

  M_opt = grover_optimize(target_state_vector)

  grover_evaluate(end_loop_states, M_opt)
\end{minted}
\caption{Main function for Grover's entanglement study}
\label{lst:grover_func}
\end{listing}

  \subsubsection{Execution} 
  \label{sub:grover_execution}

The function \py{grover_run} given in Listing~\ref{lst:grover_run} takes as
input the target state and returns a list of states composed of the states at
the end of each loop iteration. 

\begin{listing}
\begin{minted}{python}
def grover_run(target_state_vector):
  layers, k_opt = grover_layers_kopt(target_state_vector)
  N = len(target_state_vector)
  V0 = vector([0, 1] + [0]*(2*N-2))

  states = run(layers, V0)
  end_loop_states = states[0]
  for i in range(k_opt):
    end_loop_states.append(states[2*i+1])

  return end_loop_states
\end{minted}
\caption{Function running Grover's algorithm}
\label{lst:grover_run}
\end{listing}

This function operates in two steps. The first step is to build the circuit for
Grover algorithm, which is achieved by the function \py{grover_layers_kopt}. The
circuit format is a list of layers: each layer being a list of matrices (all the
operations performed at a given time) and each matrix representing an operation
performed on one or more wires. For example, if \py{H} is the Hadamard matrix,
\py{I2} and \py{I4} are the identity matrix (in dimensions 2 and 4) and \py{X}
is the first Pauli operator, then the circuit in 
Figure~\ref{fig:circ_py_example} is represented by the list \py{[[H,I4],
[X,X,I2], [I4,H], [H,H,H]]}.

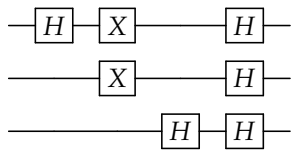
\begin{figure}[hbt!]
$$
\Qcircuit @C=1em @R=.7em {
 & \gate{H} & \gate{X}  & \qw       & \gate{H}  & \qw\\
 & \qw      & \gate{X}  & \qw       & \gate{H}  & \qw\\
 & \qw      & \qw       & \gate{H}  & \gate{H}  & \qw\\
}
$$
\caption{Example for the circuit formalism in \py{grover_ent}}
\label{fig:circ_py_example}
\end{figure}

The next step is to run the circuit, this is achieved by \py{run} which returns
the list of the states after each layer. The function \py{run} takes as input
the circuit (\py{layers}) and the initial state (\py{V0}). This function both
allows us to separate syntax and semantics, and is reusable in any future
context involving circuits. 

The \py{for}-loop then filters out all the intermediate states which are not at
the end of a loop iteration. For example, if we consider Grover's algorithm on
three qubits shown in Figure~\ref{fig:loop_count}, we would have the first state 
$\ket{\varphi_0}$, and the states $\ket{\varphi_3}$ and $\ket{\varphi_5}$ in
\py{end_loop_states}.

\begin{figure}[hbt!]
$$
\Qcircuit @C=1em @R=.7em {
\lstick{\ket{0}} & \multigate{3}{H^{\otimes n+1}} 
  & \multigate{3}{U_f} & \multigate{2}{\mathcal{D}} 
  & \multigate{3}{U_f} & \multigate{2}{\mathcal{D}} 
  & \meter & \cw \\
\lstick{\ket{0}} & \ghost{H^{\otimes n+1}} 
  & \ghost{U_f} & \ghost{\mathcal{D}}
  & \ghost{U_f} & \ghost{\mathcal{D}} 
  & \qw & \qw \\
\lstick{\ket{0}} & \ghost{H^{\otimes n+1}} 
  & \ghost{U_f} & \ghost{\mathcal{D}}
  & \ghost{U_f} & \ghost{\mathcal{D}}
  & \qw & \qw \\
\lstick{\ket{1}} & \ghost{H^{\otimes n+1}} 
  & \ghost{U_f} & \qw
  & \ghost{U_f} & \qw
  & \qw & \qw \\
\rstick{\ket{\varphi_{0}}}&&
\lstick{\scriptscriptstyle\ket{\varphi_{1}}}&
\lstick{\scriptscriptstyle\ket{\varphi_{2}}}&\lstick{\ket{\varphi_{3}}}&
\lstick{\scriptscriptstyle\ket{\varphi_{4}}}&\lstick{\ket{\varphi_{5}}}
\gategroup{1}{3}{4}{4}{.7em}{--}
\gategroup{1}{5}{4}{6}{.7em}{--}
}
$$
\caption{End loop counting example}
\label{fig:loop_count}
\end{figure}
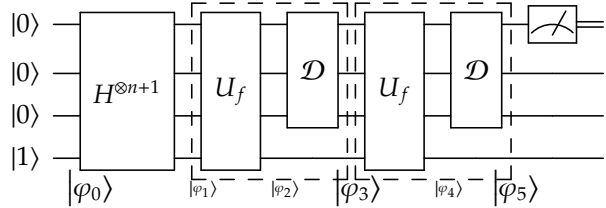

This implementation of the simulation of Grover's algorithm has its limits
though. It is computationally heavy to multiply matrices beyond a certain number
of qubits. To push it a little further, we used another implementation for
Grover's algorithm, less versatile but more efficient. This method is presented
in Listing~\ref{lst:grover_run_arti}. In this case, two important differences
are the fact that there is no more use for the ancilla qubit (the last wire in
the circuit definition of Grover's algorithm, see Figure~\ref{fig:grover_circ}),
which divides by two the number of elements in a state vector, and the fact that
almost no matrix multiplication is used. Indeed, the loop is now handled by
functions operating directly on the state vector. The first function is
\py{oracle_artificial}, and it only flips the correct coefficient in the running
state (this is the behavior explained in Section~\ref{sub:grover_algorithm}).
The second function  \py{diffusion_artificial} performs the inversion about the
mean.

\begin{listing}
\begin{minted}{python}
def grover_run(target_state_vector):
  N = len(target_state_vector)
  n = log(N)/log(2)
  k_opt = round((pi/4)*sqrt(N))
  H = matrix(field, [[1, 1], 
                    [1, -1]])/sqrt(2)
  hadamard_layer = kronecker_power(H, n)

  V0 = vector([1]+[0]*(N-1))

  V = hadamard_layer * V0
  end_loop_states = [V]

  for k in range(k_opt):
    V = oracle_artificial(target_state_vector, V)
    V = diffusion_artificial(V)
    end_loop_states.append(V)

  return end_loop_states
\end{minted}
\caption{Optimized implementation of Grover's algorithm}
\label{lst:grover_run_arti}
\end{listing} 



  \subsubsection{Optimization} 
  \label{sub:grover_optimization}

The \py{grover_optimize} function shown in Listing~\ref{lst:grover_opti}
computes an approximation of an optimal Mermin operator, as explained in
Section~\ref{sub:method_grover}. The Mermin operator $M_n$ is an implicit
function of $(\alpha, \beta, \delta,  \alpha', \beta', \delta')$, here
implemented as \py{(a,b,c,m,p,q)}. Because of this, optimizing the Mermin
operator is finding the optimal $(\alpha, \beta,  \delta, \alpha', \beta',
\delta')$ for our Mermin evaluation.

\begin{listing}
\begin{minted}{python}
def grover_optimize(target_state):
  n = log(len(target_state))/log(2)
  plus = vector([1,1])/sqrt(2)
  plus_n = kronecker_power(plus, n)
  phi = (target_state + plus_n).normalized()

  def M_phi(a,b,c,m,p,q):
    return M_eval(a,b,c,m,p,q, phi)

  (a,b,c,m,p,q),v = optimize(M_phi, (1,1,1,1,1,1), 5, 10**(-2), 10**2)

  return M_from_coef(n,a,b,c,m,p,q)
\end{minted}
\caption{Optimization function for Grover's algorithm}
\label{lst:grover_opti}
\end{listing}

To optimize the Mermin operator, first the state $\ket{\varphi_{ent}} =
(\ket{\bf x_0} + \ket{+}^{\otimes n})/K$ (with $K$ the normalizing factor) is
computed and stored in \py{phi}, then $f_{M_n}$ represented by \py{M_eval} is
used to define $f_{M_n}(\ket{ \varphi_{ent}})$ as \py{M_phi}. Note that in the
mathematical notations, $f_{M_n}(\ket{ \varphi_{ent}})$ is an implicit function
of $(\alpha, \beta, \delta, \alpha', \beta', \delta')$. This implicit relation
is made explicit as \py{M_phi} is a function of \py{(a,b,c,m,p,q)}.

The \py{optimize} function takes as input a function (here \py{M_phi}), a first
point to start the optimization from (here \py{(1,1,1,1,1,1)}), the step sizes
bounds (here \py{step_init}=$5$ and \py{step_min}$=10^{-2}$) and a maximum
number of iterations on a single step (here \py{iter_max}$=10^2$).

The optimization function proceeds with a random walk. It iterates until it
finds a local maximum (for all points $p$ in a neighborhood around the point
found $p_{opt}$, their evaluation by the function given as the first parameter
is less than the evaluation of the point found $f(p) \leq f(p_{opt})$). To find
this optimum, the process starts from an arbitrary point (given as an argument)
and at each step, an exploration of the space is done around the current point
until the evaluation on the argument function increases. If an increase cannot
be found before \py{iter_max}, the step size is reduced, otherwise, the same
step is repeated with the same step size, the function ends with the step size
reaches \py{step_min}.

\begin{remark}
This optimization can be expensive, so to speed up the calculation, a
memoization step is hidden here: if \py{(a,b,c,m,p,q)} has already been computed
for \py{target_state}, this result has been stored on disk at this point and is
now loaded.
\end{remark}


  \subsubsection{Evaluation} 
  \label{sub:grover_evaluation}

The function \py{grover_evaluate} shown in the Listing~\ref{lst:grover_eval} is
the simplest of the three: it computes $f_{M_n}(\ket{\varphi_k}) = 
\expval{M_n}{\varphi_k}$ for each $\ket{\varphi_k}$ in the \py{end_loop_states}
list with $M_n$ here being \py{M_opt}, and prints them.

\begin{listing}
\begin{minted}{python}
def grover_evaluate(end_loop_states, M_opt):
  for state in end_loop_states:
    print((state.transpose().conjugate()*M_opt*state))
\end{minted}
\caption{Evaluation function for Grover's algorithm}
\label{lst:grover_eval}
\end{listing}

To overview the code as a whole, we can show the link with  
Figure~\ref{fig:mermin_experiment}. In this case each graph has been obtained by
using a code line such as in Listing~\ref{lst:code_exec} (here
\py{target_state_ket_string_to_vector} is a function used to convert a string of
a specific format into a vector, in this case the vector is $\ket{0000}$). So,
for four qubits, we set the target state as $\ket{0000}$, for five qubits as 
$\ket{00000}$, and so on. This is enough for symmetry reasons (searching for $
\ket{1001}$ instead of $\ket{0000}$ yields similar results).

\begin{listing}
\begin{minted}{python}
>>> grover(target_state_ket_string_to_vector("0000"))
0.173154027401573
1.01189404012534
-0.469906068136016
\end{minted}
\caption{Mermin evaluation in Grover algorithm example}
\label{lst:code_exec}
\end{listing}



\subsection{Quantum Fourier Transform implementation} 
\label{sub:quantum_fourier_transform_implementation}

For the QFT, the main function \py{qft} is reproduced in
Listing~\ref{lst:qft_func}. The parameter \py{state} is the state ran through
the QFT, generally a periodic state $\ket{\varphi^{l,r}}$ generated by the
function  \py{periodic_state} (Listing~\ref{lst:periodic_generation}). The
function  \py{qft} first calls an implementation \py{qft_run} of the QFT,
detailed in Section~\ref{sub:qft_execution}, and stores the computed states in
the list  \py{states}. Then the states are directly evaluated. The important
difference compared to Grover's algorithm implementation is the fact that we are
not using a separate optimization step. Indeed, since we are not running along a
known straight path, it makes it impossible to use a single optimized Mermin
operator. Because of this, the optimization process is included in the
evaluation process: each evaluation requires an optimization. The evaluation
process is thus performed by the function \py{qft_evaluate}
(Section~\ref{sub:qft_evaluation}), printing the evaluation as well.

\begin{listing}
\begin{minted}{python}
def qft_main(state):
  states = qft_run(state)
  return qft_evaluate(states)
\end{minted}
\caption{Main function for QFT entanglement study}
\label{lst:qft_func}
\end{listing}

\begin{listing}
\begin{minted}{python}
def periodic_state(l,r,nWires):
  N = 2**nWires
  result = vector(N)
  for i in range(ceil((N-l)/r)):
    result[l+i*r] = 1
  return result.normalized()
\end{minted}
\caption{Function used to generate the periodic state $\ket{\varphi^{l,r}}$}
\label{lst:periodic_generation}
\end{listing}

\subsubsection{Execution}
\label{sub:qft_execution}

The function \py{qft_run} (Listing~\ref{lst:qft_exec}) uses the same circuit
format as \py{grover_run} presented in Section~\ref{sub:grover_execution}. This
circuit is built by \py{qft_layers} (Listing~\ref{lst:qft_layers}) and run by
\py{run}. In this case however, the states do not need to be filtered, resulting
in an almost trivial  \py{qft_run} function.

\begin{listing}
\begin{minted}{python}
def qft_run(state):
  layers = qft_layers(state)
  states, _ = run(layers, state)
  return states
\end{minted}
\caption{Function running the QFT}
\label{lst:qft_exec}
\end{listing}

The \py{qft_layers} function uses two functions not detailed here. \py{swap}
returns a matrix corresponding to the swap of two wires \py{wire1} and 
\py{wire2} and the identity on the other wires concerned. The \py{R} method
returns the controlled rotation of angle $e^{\frac{2i\pi}{2^k}}$, with the
rotation being performed on the wire \py{target} controlled by the wire 
\py{control}. The two matrices built by these function have a size \py{2**size}.
With these two functions, \py{qft_layers} builds the circuit for the QFT using
\py{R} on the whole width of the circuit when a rotation is needed and using
\py{swap} only at the end to build the global swap (in fact, \py{swap} is also
used in \py{R} and that is the reason why this implementation of swap on two
wire have been chosen instead of a more general arbitrary permutation gate).

\begin{listing}
\begin{minted}{python}
def qft_layers(state):
  def swap(wire1,wire2,size):
    ...
  def R(k,target,control,size):
    ...
  H = matrix(field, [[1,  1], 
                     [1, -1]])/sqrt(2)
  I2 = matrix.identity(field, 2)
  nWires = log(len(state))/log(2)
  layers = []

  for wire in range(nWires):
    layers.append([I2]*wire + [H] + [I2]*(nWires-wire-1))
    for k in range(2, nWires-(wire-1)):
      layers.append([R(k, wire, k+(wire-1), nWires)])

  global_swap = matrix.identity(field, 2**nWires)
  for wire in range(nWires/2):
    global_swap *= swap(wire, nWires-1-wire, nWires)
  layers.append([global_swap])

  return layers
\end{minted}
\caption{Function building the circuit of the QFT}
\label{lst:qft_layers}
\end{listing}


\subsubsection{Evaluation}
\label{sub:qft_evaluation}

In this case again, the evaluation is conceptually simpler than in Grover's
algorithm. Indeed, since the optimization needs to be performed for each
evaluation, the result printed at each step is simply the optimal point reached
by the \py{optimize} function (the same as described in
Section~\ref{sub:grover_optimization}). In this case, a notable difference in
the usage of \py{optimize} is the presence of \py{3*n*2} coefficients. Indeed,
this time, we do not want a trend for the evaluation's evolution and a "good
enough" $M_n$, we need the true optimal $M_n$ (or as least as optimal as
possible). This means that we do not stand satisfied by the constant $a_n =
\alpha X + \beta Y + \delta Z$ but we have $\alpha$, $\beta$ and $\delta$
variable as explained in \ref{sub:method_qft} (where they become $(\alpha_i)_{1
\leq i \leq 6n}$).

Because of this, the function \py{M_func} (Listing~\ref{lst:qft_eval}) we
optimize is now calling \py{M_eval_all} instead of \py{M_eval}. The difference
is that \py{M_eval} took only $3 \times 2$ coefficients to compute $M_n$ with
$a_i = \alpha X + \beta Y + \delta Z$ whereas this time the $a_i$'s are variable
thus \py{M_eval_all} takes as argument \py{_a_coefs} and  \py{_a_prime_coefs}
two lists of triplets (each triplet encoding one $a_i$). This is the reason why
we need to go through \py{coefficients_packing}:  the \py{optimize} function
needs a flat list of reals to feed into the optimized function, so to accomplish
that, the arguments of \py{M_func} is a flat list and is packed in the proper
shape by \py{coefficients_packing} ($a_i = $ \py{_a_coefs[i][0]}$X +
$\py{_a_coefs[i][1]}$Y + $\py{_a_coefs[i][2]}$Z$).

\begin{listing}
\begin{minted}{python}
def qft_evaluate(states):
  n = log(len(states[0]))/log(2)
  for state in states:
	  rho = matrix(state).transpose()*matrix(state)

	  def M_func(_a_a_prime_coefs):
	    _a_coefs, _a_prime_coefs = coefficients_packing(_a_a_prime_coefs)
	    return M_eval_all(n, _a_coefs, _a_prime_coefs, rho)

	  _,value = optimize(M_func, [1]*3*n*2, 5, 10**(-2), 10**2)

	  print value
\end{minted}
\caption{Evaluation function for the QFT}
\label{lst:qft_eval}
\end{listing}


\subsection{Implementation recap} 
\label{sub:implementation_recap}

Finally, to conclude this section, we recall the functions reusable in a general
context, the \py{run} function can be used for general purpose quantum circuit
simulation and the Mermin evaluation process can be used for arbitrary state
entanglement evaluation. An issue previously mentioned was the correctness
between the process and the simulation, and here this issue is tackled by
structured and clear code. This structure also helps the code to be more
modular, for instance, if the user wants to change the optimization method for
more speed or precision, it can be easily achieved.

\begin{remark}
Note that the actual functions have more parameters that are ignored here for
simplicity's sake. For example, each function has a verbose mode, to display
more information about its run.
\end{remark}




\section{Conclusion} 
\label{sec:conclusion}

With these experiments, we showed that evaluation with Mermin polynomials is a
valuable tool to study entanglement within quantum algorithms. The study of
Grover's algorithm showed us that the Mermin evaluation can be used to check
properties like non-locality and evolution of entanglement during the execution
of the algorithm. In our study of the QFT algorithm we  showed that the Mermin
evaluation can sometimes be compared to the evaluation of an algebraic
invariant, such as the Cayley hyperdeterminant, but not consistently. 

This possibility of ``property checking'' is promising as an attack point on the
problem of quantum program verification. Indeed, in both cases studied in this
article, the Mermin evaluation corresponds to an experimental measurement that
could be performed on a quantum computer device. See for instance~\cite{AL16}
for examples  of Mermin evaluation of a $5$-qubit computer. So, in addition to
studying more algorithms, it may be interesting to use the Mermin evaluation in
arbitrary state checking in true quantum computers in a near future.


\section*{Acknowledgments} 
\label{sec:acknowledgments}
\addcontentsline{toc}{section}{Acknowledgement}

This project is supported by the French Investissements d'Avenir program,
project ISITE-BFC (contract ANR-15-IDEX-03). The computations have been
performed on the supercomputer facilities of the Mésocentre de calcul de
Franche-Comté.


\bibliographystyle{alpha}
\bibliography{library.bib}

\newpage

\appendix
\section{Explicit states for Grover's algorithm} 
\label{sec:explicit_states_for_grover_s_algorithm}

\begin{proposition}
\label{prop:explicit_state_grover}

\cite[Observation 1]{HJN16} The state $\ket{\varphi_k}$ after $k$ iterations of
Grover's algorithm can be written as follows:

\begin{equation}
\label{eq:explicit_states_appendix}
\ket{\varphi_k} = \tilde{\alpha}_k \sum_{{\bf x}\in S} \ket{\bf x} + 
\tilde{\beta}_k \ket{+}^{\otimes n}
\end{equation}

\noindent with $\tilde{\alpha}_k=\dfrac{\cos(\frac{2k+1}{2} \theta)}{\sqrt{|S|}}
- \dfrac{\sin(\frac{2k + 1}{2}\theta)}{\sqrt{N - |S|}}$ and
$\tilde{\beta}_k = 2^{n/2} \dfrac{\sin(\frac{2k + 1}{2}\theta)}{\sqrt{N-|S|}}$.
\end{proposition}

\begin{proof}
With $\ket{\varphi_0} = \ket{+}^{\otimes n}$, we can write:

\begin{equation*}
\ket{\varphi_k} = \mathcal{L}^k\ket{\varphi_0} = \dfrac{a_k}{\sqrt{|S|}}\sum_
{{\bf x}\in S} \ket{\bf x}+ \dfrac{b_k}{\sqrt{N-|S|}}\sum_{{\bf x}\notin S} 
\ket{\bf x}
\end{equation*}

\noindent where $\mathcal{L}$ is the loop (oracle and diffusion operator) in
Grover's algorithm.

The oracle is a reflection about $(\sum_{{\bf x}\in S}\ket{\bf x})^\bot = \sum_{
{\bf x}\notin S} \ket{\bf x}$ and the diffusion operator is a reflection about
$\ket{+}^{\otimes n}$. The composition of these two symmetries is a rotation
whose angle $\theta$ is the double of the angle between $\sum_{{\bf x}\notin S}
\ket{\bf x}$ and $\ket{+}^{\otimes n}$. So,
$$
\begin{array}{rll}
\ket{+}^{\otimes n}     & = & \frac{1}{\sqrt{|S|}}\sin(\frac{\theta}{2})
    \sum_{{\bf x}\in S}\ket{\bf x}+\frac{1}{\sqrt{N-|S|}}\cos(\frac{\theta}{2})
    \sum_{{\bf x}\notin S} \ket{\bf x}\\
\frac{1}{\sqrt{N}} \left(\sum_{{\bf x}\in S} \ket{\bf x}+\sum_{{\bf x}\notin S} 
  \ket{\bf x}\right)    & = & \frac{1}{\sqrt{|S|}}\sin(\frac{\theta}{2})
    \sum_{{\bf x}\in S}\ket{\bf x}+\frac{1}{\sqrt{N-|S|}}\cos(\frac{\theta}{2})
    \sum_{{\bf x}\notin S} \ket{\bf x}\\
\frac{1}{\sqrt{N}}\sum_{{\bf x}\in S} \ket{\bf x}
                        & = & \frac{1}{\sqrt{|S|}}\sin(\frac{\theta}{2})
    \sum_{{\bf x}\in S} \ket{\bf x}\\
\frac{1}{\sqrt{N}}      & = & \frac{1}{\sqrt{|S|}}\sin(\frac{\theta}{2})\\
\sin(\frac{\theta}{2})  & = & \sqrt{\frac{|S|}{N}}.
\end{array}
$$
The fact that $\mathcal{L}$ is a rotation of angle $\theta$ gives
 $a_k=\sin\left(\theta_k\right)$  and $b_k = \cos\left( 
\theta_k \right)$  with $\theta_k = k \theta + \theta/2$.
 Equation~(\ref{eq:explicit_states}) then comes from $\alpha_k = 
\frac{1}{\sqrt{|S|}}\sin(\frac{2k + 1}{2} \theta)$ and $\beta_k =  
\frac{1}{\sqrt{N-|S|}}\cos(\frac{2k + 1}{2}\theta)$.

With this, we can now take  $\tilde{\alpha}_k = \alpha_k - \beta_k$  and
$\tilde{\beta}_k = 2^{n/2} \beta_k$ which gives us
\begin{equation*}
 \begin{array}{lll}
  \ket{\varphi_k} & = & \alpha_k \sum_{{\bf x}\in S} \ket{\bf x} + \beta_k 
                    \sum_{{\bf x}\notin S} \ket{\bf x}\\
                  & = & (\alpha_k - \beta_k) \sum_{{\bf x}\in S} \ket{{\bf x}} + 
                    \beta_k \sum_{{\bf
x}=0}^{N-1} \ket{\bf x}\\
                  & = & \tilde{\alpha}_k\sum_{{\bf x}\in S}\ket{{\bf x}} + 
                    \tilde{\beta}_k \ket{+}^{\otimes n}
 \end{array}
\end{equation*}
since $\ket{+}^{\otimes n} = \left(\frac{1}{\sqrt{2}}\right)^n \sum_{{\bf
x}=0}^{N-1} \ket{\bf x}$. \qed
\end{proof}

\begin{proposition}
\label{prop:grover_a_b_monotonous}
In Proposition~\ref{prop:explicit_state_grover}, $\tilde{\alpha}_k$ increases
for $k$ between 0 and $\frac{\pi}{4}\sqrt{\frac{N}{|S|}}-\frac{1}{2}$ and
$\tilde {\beta}_k$ decreases on the same interval.
\end{proposition}

\begin{proof}

The optimal number of iterations of the loop $\mathcal{L}$ in Grover's algorithm
is the smallest value $k_{opt}$ of $k$ such that $a_k = 1$, \emph{i.e.},
$\theta_{k_{opt}} = \pi/2$. With $|S|\ll N$,
$\sin\left(\theta/2\right)=\sqrt{|S|/N}$ gives $\theta \approx 2\sqrt{|S|/N}$
and $\theta_k \approx (2k+1)\sqrt{|S|/N}$. Finally $(2k_{opt}+1)\sqrt{|S|/N}$
 optimally approximates $\pi/2$ if $k_{opt} =
\left\lfloor \frac{\pi}{4}\sqrt{\frac{N}{|S|}}-\frac{1}{2} \right\rceil = \left\lfloor
\frac{\pi}{4}\sqrt{\frac{N}{|S|}} \right\rfloor$.

Moreover, $a_k=\sin\left(\theta_k\right)$ and $\alpha_k = \frac{1}{\sqrt{|S|}}
a_k$ are increasing and $b_k = \cos\left( \theta_k \right)$ and $\beta_k =  
\frac{1}{\sqrt{N-|S|}} b_k$ are 
decreasing for $k$ from $0$ to
$\left(\frac{\pi}{4}\sqrt{\frac{N}{|S|}}-\frac{1}{2}\right)$.
From the expressions $\tilde{\alpha}_k = \alpha_k - \beta_k$ and $
\tilde{\beta}_k = 2^{n/2} \beta_k$, we get the result of the proposition.\qed

\end{proof}

\newpage
\section{Cayley hyperdeterminant \texorpdfstring{$\Delta_{2222}$}{D2222}}
\label{appendix:cayley_hyperdeterminant_D2222}

Let $\ket{\varphi}=\sum_{i,j,k,l\in\{0,1\}} a_{i,j,k,l}\ket{ijkl}$ be a
four-qubit state. The algebra of polynomial invariants for the four-qubit
Hilbert space can be generated by the four polynomials $H$, $L$, $M$ and $D$
defined as follows~\cite{LT03}:

\[H=a_{0000}a_{1111} - a_{1000}a_{0111} - a_{0100}a_{1011} + a_{1100}a_{0011}\] 
\[-a_{0010}a_{1101} + a_{1010}a_{0101} + a_{0110}a_{1001} - a_{1110}a_{0001}\]
is an invariant of degree 2.

\[ L=\left|\begin{array}{cccc}
a_{0000}&a_{0010}&a_{0001}&a_{0011}\\
a_{1000}&a_{1010}&a_{1001}&a_{1011}\\
a_{0100}&a_{0110}&a_{0101}&a_{0111}\\
a_{1100}&a_{1110}&a_{1101}&a_{1111}
\end{array}\right| \hspace{10mm}\mbox{ and }\hspace{10mm} 
M=\left|\begin{array}{cccc}
a_{0000}&a_{0001}&a_{0100}&a_{0101}\\
a_{1000}&a_{1001}&a_{1100}&a_{1101}\\
a_{0010}&a_{0011}&a_{0110}&a_{0111}\\
a_{1010}&a_{1011}&a_{1110}&a_{1111}
\end{array}\right|
\]
are two invariants of degree 4.

Consider the partial derivative 
$$b_{xt}:=\det\left(\dfrac{\partial^2 A} {\partial
y_i\partial z_j}\right)$$
of the quadrilinear form $A =\sum_{i,j,k,l\in\{0,1\}} a_{i,j,k,l} x_i y_j z_k 
t_l$ with respect to the variables $y$ and $z$. This quadratic form with
variables $x$ and $t$ can be interpreted as a bilinear form on the
three-dimensional space $\text{Sym}^2(\mathbb{C}^2)$, \emph{i.e.}, there is a
$3\times 3$ matrix $B_{xt}$ satisfying
\[
b_{xt}=[x_0^2,x_0x_1,x_1^2]~B_{xt}~\left[
\begin{array}{c}
t_0^2\\
t_0t_1\\
t_1^2 
\end{array}\right].
\]
Then $D =\det(B_{xt})$ is an invariant of degree 6. 

Let's introduce the invariant polynomials 
$$U=H^2-4(L-M), \hspace{10mm} V=12(HD-2LM),$$

$$ S=\dfrac{1}{12}(U^2-2V)  \hspace{10mm}  \mbox{ and } \hspace{10mm} 
T=\dfrac{1}{216}(U^3-3UV+216D^2).$$
\noindent  Then the Cayley hyperdeterminant
is~\cite{LT03}:
\begin{equation*}
 \Delta_{2222}=S^3-27T^2.
\end{equation*}

\end{document}

%% file: main-header.tex
\usepackage[pdfencoding=auto,
  colorlinks=true,
  linkcolor=black,
  urlcolor=black,
  citecolor=black]{hyperref}
\usepackage[T1]{fontenc}
\usepackage[utf8x]{inputenc}
\usepackage{times}
\usepackage[english]{babel}
\usepackage{url}
\usepackage{amsmath,amssymb}
\usepackage[table]{xcolor}
\usepackage{a4wide}
\usepackage{multicol}
\usepackage{graphicx}
\PrerenderUnicode{Ãl’}
\usepackage{caption}
\usepackage{subcaption}
\usepackage{float}
\usepackage{tikz,pgfplots}
\usetikzlibrary{decorations.markings}
\usetikzlibrary{shapes.geometric}
\usetikzlibrary{calc}
\pgfdeclarelayer{edgelayer}
\pgfdeclarelayer{nodelayer}
\pgfsetlayers{edgelayer,nodelayer,main}

\tikzstyle{none}=[inner sep=0pt]

\tikzstyle{rn}=[circle,fill=Red,draw=black,line width=0.8 pt]
\tikzstyle{gn}=[circle,fill=Lime,draw=black,line width=0.8 pt]
\tikzstyle{yn}=[circle,fill=Yellow,draw=black,line width=0.8 pt]
\tikzstyle{smallCirc}=[circle,fill=White,draw=black]
\tikzstyle{bigCirc}=[circle,fill=White,draw=black]
\tikzstyle{bn}=[inner sep=0pt, circle,fill=black,draw=black,line width=0.8 pt,minimum size=4]

\tikzstyle{simple}=[-,draw=black,line width=1]
\tikzstyle{tick}=[-,draw=black,postaction={decorate},decoration={markings,mark=at position .5 with {\draw (0,-0.1) -- (0,0.1);}},line width=2.000]
\tikzstyle{arrow}=[-,draw=black,postaction={decorate},decoration={markings,mark=at position .5 with {\arrow{>}}}]
\tikzstyle{arrowHead}=[->,draw=black]
\usepackage{physics}
\usepackage{pdflscape}
\usepackage{array,multirow}
\usepackage{relsize}
\usepackage{pxfonts}
\usepackage{markdown}
\usepackage{pdftexcmds}
\usepackage{qcircuit}
\usepackage{empheq}
\pgfplotsset{compat=1.14}
\pgfplotsset{groverExp/.style={ymin=-1.2, ymax=2.4, ytick={0}, grid=major, 
	xtick={0}, grid style={yticklabel=\empty}, ytick style={draw=none}}}
\pgfplotsset{qftExp/.style={ymin=.8, ymax=3, extra y ticks={1}, 
  extra tick style={grid style={red},grid=major},
  width = \textwidth, height = 4cm, xlabel = {$k$}, ylabel = {$\tilde{q}(k)$}}}
\pgfplotscreateplotcyclelist{my cycle}{%
color=blue!40!black,mark=*\\%
color=red!50!black,mark=+\\%
color=green!40!black,mark=diamond*\\%
color=orange!70!black,mark=triangle*\\%
color=purple,mark=square*\\%
color=cyan,mark=star\\%
color=yellow,mark=otimes\\%
color=red,mark=x\\%
color=black,mark=pentagon*\\%
}

\usepackage{mathtools}

\DeclarePairedDelimiter{\ceil}{\lceil}{\rceil}
\DeclarePairedDelimiter{\floor}{\lfloor}{\rfloor}

\usepackage[newfloat=true,finalizecache=false,frozencache=true]{minted}
\SetupFloatingEnvironment{listing}{placement=htb!}
\setminted{tabsize=2}
\usepackage{enumitem}

\newcommand{\py}[1]{\mintinline{python}{#1}}

\patchcmd{\abstract}{-.5em}{0em}{}{}

\usepackage{caption}
\newcommand\ifdeftostring[3]{
  \def\next{#2}%
  \ifx#1\next
    #3
  \fi
}

\ifdeftostring{\version}{qip}
  {
\smartqed
  }

\ifdeftostring{\version}{arxiv}
  {

\usepackage{amsthm}
\usepackage{authblk}
\newtheorem{definition}{Definition}

\newtheorem{example}{Example}
\newtheorem{remark}{Remark}
\newtheorem{proposition}{Proposition}
\newcommand\blfootnote[1]{%
  \begingroup
  \renewcommand\thefootnote{}\footnote{#1}%
  \addtocounter{footnote}{-1}%
  \endgroup
}
  }

%% file: resources/grover_run.tex
\begin{subfigure}{0.5\textwidth} \begin{tikzpicture} \begin{axis}[groverExp]
\addplot+[ycomb] coordinates {(1, 1) (2, 1)		   (4, 1)};
\addplot+[ycomb] coordinates {					   (3, 1)};
  \node (title) at (3.2,-.2) {$\bf x_0$};
\end{axis} \end{tikzpicture} 
\caption{}
\label{fig:superposition}
\end{subfigure}
\begin{subfigure}{0.5\textwidth} \begin{tikzpicture} \begin{axis}[groverExp]
\addplot+[ycomb] coordinates {(1, 1) (2, 1)		 	 (4, 1)};
\addplot+[ycomb] coordinates {					   (3, -1)};
  \node (title) at (3.2,-.2) {$\bf x_0$};
\end{axis} \end{tikzpicture} 
\caption{}
\label{fig:oracle}
\end{subfigure}

\begin{subfigure}{0.5\textwidth} \begin{tikzpicture}
\begin{axis}[ymin=-1.2, ymax=2.4, ytick={0}, grid=major, xtick={0}, grid style={yticklabel=\empty},
	extra y ticks={0.6}, extra tick style={yticklabel=\empty,grid style={black, dashed},grid=major}, ytick style={draw=none}]
\addplot+[ycomb] coordinates {(1, .2) (2, .2)			    (4, .2)};
\addplot+[ycomb] coordinates {					     (3, 2.2)};
\addplot+[ycomb, blue, dashed, mark=*] coordinates {(1, 1) (2, 1)	    (4, 1)};
\addplot+[ycomb, red, dashed, mark=square] coordinates {		    (3, -1)};
  \node (title) at (3.2,-.2) {$\bf x_0$};
\end{axis} \end{tikzpicture} 
\caption{}
\label{fig:invertion}
\end{subfigure}
\begin{subfigure}{0.5\textwidth} \begin{tikzpicture} \begin{axis}[groverExp]
\addplot+[ycomb] coordinates {(1, .2) (2, .2)         (4, .2)};
\addplot+[ycomb] coordinates {               (3, 2.2)};
  \node (title) at (3.2,-.2) {$\bf x_0$};
\end{axis} \end{tikzpicture} 
\caption{}
\label{fig:inverted}
\end{subfigure}

%% file: resources/HJN_fig2.tex
\begin{tikzpicture}
  \begin{pgfonlayer}{nodelayer}
\node [style=none] (0) at (-4.5,0.5) {};
\node [style=bn, label={below:$\ket{+}^{\otimes n}$}] (01) at (-3.6,-.1) {};
\node [style=none] (1) at (0.5,1.5) {};
\node [style=none] (2) at (-2,3.5) {};
\node [style=bn, label={above right:$\ket{\bf x_0}$}] (3) at (-2.5,1.5) {};
\node [style=none] (4) at (3.5,2.5) {};
\node [style=bn, label={above left:$\ket{\varphi_{\floor{k_{opt}/2}}}$}] (5) at (-3.1,0.65) {};
\node [style=none, label={right:$X$}] (5) at (2.3,1.5) {};
\node [style=none] (BoundingBoxFixUp) at (0,3.6) {};
\node [style=none] (BoundingBoxFixleft) at (3.6,0) {};
\node [style=none] (BoundingBoxFixDown) at (0,-1.1) {};
  \end{pgfonlayer}
  \begin{pgfonlayer}{edgelayer}
\draw[/tikz/overlay] (0.center) .. controls (-1,-2) and (0.5,-1) .. (1.center);
\draw[/tikz/overlay] (1.center) .. controls (0.5,3.5) and (-0.5,3.5) .. (2.center);
\draw[/tikz/overlay] (2.center) .. controls (-3.5,3.5) and (-3.5,3) .. (3.center);
\draw[/tikz/overlay] (3.center) .. controls (0,-1.5) and (1.5,-2.5) .. (4.center);
\draw[dotted, thick, /tikz/overlay] (3.center) -- (01.center);
  \end{pgfonlayer}
\end{tikzpicture}

%% file: resources/mermin_for_grover.tex
\begin{tikzpicture} 
\begin{axis}[extra y ticks={1}, 
    extra tick style={grid style={red},grid=major},
	xtick={0,5,...,100}, 
	width = 0.9\textwidth, height = 5cm,
	xlabel = {Number of iterations}, ylabel = {Mermin evaluation},
   	legend style ={ at={(1.03,1)}, 
        anchor=north west, draw=black, 
        fill=white,align=left},
    cycle list name=my cycle] 

	\addplot table[x=iteration, y=intricationValue, col sep=comma]{resources/data/grover_experience_results/4qbit.csv};
	\addlegendentry{4 qubits}
	\addplot table[x=iteration, y=intricationValue, col sep=comma]{resources/data/grover_experience_results/5qbit.csv};
	\addlegendentry{5 qubits}
	\addplot table[x=iteration, y=intricationValue, col sep=comma]{resources/data/grover_experience_results/6qbit.csv};
	\addlegendentry{6 qubits}
	\addplot table[x=iteration, y=intricationValue, col sep=comma]{resources/data/grover_experience_results/7qbit.csv};
	\addlegendentry{7 qubits}
	\addplot table[x=iteration, y=intricationValue, col sep=comma]{resources/data/grover_experience_results/8qbit.csv};
	\addlegendentry{8 qubits}
	\addplot table[x=iteration, y=intricationValue, col sep=comma]{resources/data/grover_experience_results/9qbit.csv};
	\addlegendentry{9 qubits}
	\addplot table[x=iteration, y=intricationValue, col sep=comma]{resources/data/grover_experience_results/10qbit.csv};
	\addlegendentry{10 qubits}
	\addplot table[x=iteration, y=intricationValue, col sep=comma]{resources/data/grover_experience_results/11qbit.csv};
	\addlegendentry{11 qubits}
	\addplot table[x=iteration, y=intricationValue, col sep=comma]{resources/data/grover_experience_results/12qbit.csv};
	\addlegendentry{12 qubits}
\end{axis} 
\end{tikzpicture}

%% file: resources/Grover_BOF.tex
\begin{tikzpicture} 
\begin{axis} [extra y ticks={1}, 
    extra tick style={grid style={red},grid=major},
	xtick={0,1,...,100}, 
	width = \textwidth, height = 5cm,
	xlabel = {Number of iterations}, ylabel = {Mermin evaluation},
   	legend style ={ at={(1.03,1)}, 
        anchor=north west, draw=black, 
        fill=white,align=left},
    cycle list name=my cycle] 

	\addplot table[x=iteration, y=intricationValue, col sep=comma]{resources/data/Grover_BOF.csv};
\end{axis} 
\end{tikzpicture}

%% file: resources/max_exp-theo.tex
\begin{tikzpicture} 
\begin{axis} [extra y ticks={0}, 
    extra tick style={grid style={black},grid=major},
	width = 0.7\textwidth, height = 4cm,
	xlabel = {Number of qubits}, ylabel = {Mermin evaluation},
   	legend style ={ at={(1.03,1)}, 
        anchor=north west, draw=black, 
        fill=white,align=left},
    cycle list name=my cycle] 

	\addplot table[x=n, y=exp, col sep=comma]{resources/data/max_exp-theo.csv};
	\addlegendentry{Experimental results}

	\addplot table[x=n, y=theo, col sep=comma]{resources/data/max_exp-theo.csv};
	\addlegendentry{Theoretical upper boundary}
\end{axis} 
\end{tikzpicture}

%% file: resources/qft_hyperdet_comparison.tex
\begin{subfigure}{0.5\textwidth} \begin{tikzpicture} 
\begin{axis}[extra y ticks={0}, 
    extra tick style={grid style={black},grid=major},
    width = \textwidth, height = 4cm, 
    xlabel = {$k$}, ylabel = {$\Delta_{2222}$ evaluation}
]
  \addplot table[x=iteration, y=intricationValue, col sep=comma]{resources/data/period_9-1_hyperdet.csv};
\end{axis} 
\end{tikzpicture}
\end{subfigure}
\begin{subfigure}{0.5\textwidth} \begin{tikzpicture} 
\begin{axis}[qftExp] 
  \addplot table[x=iteration, y=intricationValue, col sep=comma]{resources/data/qft_experience_results/period_9-1.csv};
\end{axis} 
\end{tikzpicture}
\end{subfigure}

%% file: resources/qft_res_hyperdet_not_null.tex
\begin{subfigure}{0.5\textwidth} \begin{tikzpicture} 
\begin{axis}[qftExp] 
  \addplot table[x=iteration, y=intricationValue, col sep=comma]{resources/data/qft_experience_results/period_1-3.csv};
\end{axis} 
\end{tikzpicture}
\caption{$(l,r)=(1,3)$}
\label{fig:qft-l1r3}
\end{subfigure}
\begin{subfigure}{0.5\textwidth} \begin{tikzpicture} 
\begin{axis}[qftExp] 
  \addplot table[x=iteration, y=intricationValue, col sep=comma]{resources/data/qft_experience_results/period_2-3.csv};
\end{axis} 
\end{tikzpicture}
\caption{$(l,r)=(2,3)$}
\label{fig:qft-l2r3}
\end{subfigure}

%% file: resources/qft_res_hyperdet_not_null_after.tex
\begin{subfigure}{0.5\textwidth} \begin{tikzpicture}
\begin{axis}[qftExp] 
  \addplot table[x=iteration, y=intricationValue, col sep=comma]{resources/data/qft_experience_results/period_0-3.csv};
\end{axis}
\end{tikzpicture} 
\caption{$(l,r)=(0,3)$}
\label{fig:qft-l0r3}
\end{subfigure}
\begin{subfigure}{0.5\textwidth} \begin{tikzpicture} 
\begin{axis}[qftExp] 
  \addplot table[x=iteration, y=intricationValue, col sep=comma]{resources/data/qft_experience_results/period_5-3.csv};
\end{axis} 
\end{tikzpicture} 
\caption{$(l,r)=(5,3)$}
\label{fig:qft-l5r3}
\end{subfigure}
\begin{subfigure}{0.5\textwidth} \begin{tikzpicture} 
\begin{axis}[qftExp] 
  \addplot table[x=iteration, y=intricationValue, col sep=comma]{resources/data/qft_experience_results/period_11-1.csv};
\end{axis} 
\end{tikzpicture} 
\caption{$(l,r)=(11,1)$}
\label{fig:qft-l11r1}
\end{subfigure}
\begin{subfigure}{0.5\textwidth} \begin{tikzpicture} 
\begin{axis}[qftExp] 
  \addplot table[x=iteration, y=intricationValue, col sep=comma]{resources/data/qft_experience_results/period_9-1.csv};
\end{axis} 
\end{tikzpicture}
\caption{$(l,r)=(9,1)$}
\label{fig:qft-l9r1}
\end{subfigure}

%% file: resources/qft_res_hyperdet_null.tex
\begin{subfigure}{0.5\textwidth} \begin{tikzpicture} 
\begin{axis}[qftExp] 
  \addplot table[x=iteration, y=intricationValue, col sep=comma]{resources/data/qft_experience_results/period_2-2.csv};
\end{axis} 
\end{tikzpicture} 
\caption{$(l,r)=(2,2)$}
\label{fig:qft-l2r2}
\end{subfigure}
\begin{subfigure}{0.5\textwidth} \begin{tikzpicture} 
\begin{axis}[qftExp] 
  \addplot table[x=iteration, y=intricationValue, col sep=comma]{resources/data/qft_experience_results/period_0-11.csv};
\end{axis}
\end{tikzpicture} 
\caption{$(l,r)=(0,11)$}
\label{fig:qft-l0r11}
\end{subfigure}
\begin{subfigure}{0.5\textwidth} \begin{tikzpicture} 
\begin{axis}[qftExp] 
  \addplot table[x=iteration, y=intricationValue, col sep=comma]{resources/data/qft_experience_results/period_0-15.csv};
\end{axis}
\end{tikzpicture} 
\caption{$(l,r)=(0,15)$}
\label{fig:qft-l0r15}
\end{subfigure}
\begin{subfigure}{0.5\textwidth} \begin{tikzpicture} 
\begin{axis}[qftExp] 
  \addplot table[x=iteration, y=intricationValue, col sep=comma]{resources/data/qft_experience_results/period_1-1.csv};
\end{axis} 
\end{tikzpicture} 
\caption{$(l,r)=(1,1)$}
\label{fig:qft-l1r1}
\end{subfigure}